\documentclass[10pt,draft,reqno]{amsart}

\usepackage{color}

\usepackage{enumitem}

\usepackage{bm}
\usepackage[T1]{fontenc}
\usepackage{graphics}
\usepackage[colorlinks,linkcolor=blue]{hyperref}
\usepackage{mathtools}
\usepackage{multicol}
\usepackage{systeme}

\newcommand{\BB}[1]{{\mathbb #1}}
\newcommand{\B}[1]{{\mathbf #1}}
\newcommand{\C}[1]{{\mathcal #1}}

\newcommand{\OP}{\operatorname}

     \makeatletter
     \def\section{\@startsection{section}{1}%
     \z@{.7\linespacing\@plus\linespacing}{.5\linespacing}%
     {\bfseries
     \centering
     }}
     \def\@secnumfont{\bfseries}
     \makeatother
\newtheorem{theorem}{Theorem}[section]
\newtheorem{lemma}[theorem]{Lemma}
\newtheorem{proposition}[theorem]{Proposition}
\newtheorem{corollary}[theorem]{Corollary}
\theoremstyle{definition}
\newtheorem{definition}[theorem]{Definition}
\newtheorem{example}[theorem]{Example}
\theoremstyle{remark}
\newtheorem{remark}[theorem]{Remark}
\numberwithin{equation}{section}
\begin{document}

\title[Multi-asset binomial market]
{Pricing multi-asset contingent claims in a multi-dimensional binomial market}
\author{Jarek K\k{e}dra}
\address{Jarek K\k{e}dra: 
University of Aberdeen, Department of Mathematics, Fraser Noble Building, AB24 3UE Aberdeen, Scotland, UK}
\email{kedra@abdn.ac.uk}
\author{Assaf Libman}
\address{Assaf Libman: 
University of Aberdeen, Department of Mathematics, Fraser Noble Building, AB24 3UE Aberdeen, Scotland, UK}
\email{a.libman@abdn.ac.uk}
\author{Victoria Steblovskaya}
\address{Victoria Steblovskaya: Bentley University, 175 Forest Street, Waltham, MA 02452, USA}
\email{vsteblovskay@bentley.edu}


\begin{abstract}
We consider an incomplete multi-dimensional binomial market and a multi-asset European type contingent claim in it. For a general multi-asset contingent claim, we build straightforward algorithms that return the boundaries of a no-arbitrage contingent claim price interval. These algorithms are replaced with explicit formulas for a wide class of contingent claims (both path-independent and path-dependent). This simplification is possible due to the following remarkable fact: for this class of contingent claims, an extremal multi-step martingale
measure is a power of the corresponding single-step extremal martingale measure for which an explicit formula is provided. Our results apply, for example, to European basket call and put options and Asian arithmetic average options.
\end{abstract}

\maketitle

\section{Introduction}\label{S:intro}

Various extensions of the classical single-dimensional binomial model of Cox, Ross, and Rubinstein \cite{MR3821653} have been studied in connection with financial applications for the last several decades. One line of research originating from \cite{Boyle} (see e.g. \cite{BoyleEtAl,MR2378982,KaRi,MR2506703, MR3175871} and references therein) is devoted to the problem of pricing multi-asset contingent claims (options) in complex continuous time markets by means of binomial approximations. The continuous time multi-dimensional stochastic processes that model stock price dynamics are approximated with suitable multi-dimensional binomial discrete time processes. So the multi-dimensional binomial model in this approach serves only as a convenient computational tool.

Another line of research (see \cite{MR1611839,MR2283328} and references therein) explores the multi-dimensional binomial model as an incomplete market model where pricing and hedging of multi-asset financial derivatives is studied. In an incomplete market, there is no unique no-arbitrage price for an option. Instead, there is an interval of no-arbitrage option prices, and each price in that interval is obtained as a discounted expected option pay-off with respect to a certain martingale measure. The boundaries of the no-arbitrage option price interval are obtained by means of the extremal martingale measures. Identifying these measures and the boundaries of the no-arbitrage option price interval is an important theoretical problem which has direct applications in practice.

Motoczy\'nski and Stettner \cite{MR1611839} studies the binomial market with $d$ risky assets whose price processes are independent of each other. Using a geometric approach, the authors obtain a general formula for the upper bound of an interval of no-arbitrage prices for a
multi-asset option in terms of an iteration of a certain transformation within the space of Borel real valued measurable
functions on $\BB R^d$. This general formula is further developed analytically for the case of a two-asset single-step binomial market. Explicit formulas for the maximal martingale measure, as well as the upper bound of the no-arbitrage option price interval are obtained. 
The authors show that their results can be extended to the case of bounded stock price ratios and a convex pay-off function.

In \cite{MR2283328}, the two-asset binomial market is studied further and extended to the multi-step case. No assumptions are made regarding the joint distribution of the stock price ratios. For a large class of the two-asset options,  explicit formulas are obtained for both, the upper and the lower bounds of the no-arbitrage price interval. The results are extended to the case where the stock price ratios are distributed over a closed rectangle and the pay-off function is convex.

Further multi-dimensional extensions of the CRR model can be found in e.g  \cite{MR1725751} and \cite{MR1987320}. In \cite{MR1725751}, a version of the multi-dimensional CRR model with no riskless asset is considered. A special attention is paid to the no-arbitrage assumptions for this market. 

In \cite{MR1987320}, the authors consider pricing multi-asset options in discrete time Markovian models. Using the dynamic programming approach, they obtain the upper bounds of the no-arbitrage option price intervals (in their terminology, superprices) for European and American type options. The two-dimensional single-step CRR model occurs as an example. 

The present paper further develops the studies of \cite{MR1611839,MR2283328}.   
We consider a discrete time $n$-step market model with $m$ risky assets (stocks). Each stock price 
follows a binomial model. No assumptions are made on the joint distribution of the stock price processes. In this market, we consider both path-dependent and path-independent multi-asset European contingent claims (options). We compute bounds of the no-arbitrage option price interval at every time step via an iterated linear programming algorithm. 

The above algorithm is significantly improved and in many cases replaced with closed form formulas for a special class of contingent claims: fibrewise supermodular contingent claims (see Definition \ref{D:f-supermodular}). This class is sufficiently large. It includes many commonly used multi-asset options, such as e.g. European basket call and put options as well as  
Asian (path-dependent) basket options based on arithmetic average. 

The computations require a careful study of the set of appropriate martingale measures. 
Our main theoretical result (Theorem \ref{T:assaf}) states that for a fibrewise supermodular contingent claim, the maximal and minimal martingale measures that deliver the boundaries of the no-arbitrage option price interval are product measures of the form $\B P^{\otimes n}=\B P\otimes \B P\otimes \cdots \otimes \B P$, where $\B P$ is an appropriate 
single-step extremal martingale measure. Moreover, the single-step extremal martingale measures can be described explicitly (Section \ref{S:explicit}). This remarkable fact allows writing down explicit formulas for the boundaries of the no-arbitrage option price intervals.

\begin{remark}
	The proof of Theorem \ref{T:assaf} is an elaboration of the well known result of the combinatorial
	optimization which states that supermodular set functions when restricted to convex 
	polytopes are maximized on a special vertex \cite{MR717403,MR3154633,MR468177,MR548704}. 
	We present detailed arguments in the Appendix, which is somewhat more
	technical than the rest of the paper.  
\end{remark}

\begin{remark}
The multi-dimensional CRR market is incomplete even in the simplest case of two assets. It is interesting to observe that the multi-dimensional Black-Scholes model  which represents the limiting case of this model when the number of time steps tends to infinity is complete \cite{MR1421066}. This provides a motivation to study this discrete time model since it reveals interesting effects that are not observed in the continuous time setting.  
\end{remark}

\begin{remark}
As an additional motivation for studying the multi-dimensional binomial market,
we refer to the line of research 
\cite{JKS-risk,MR3484546,MR3612260}. These studies are
devoted to the extension of a classical single-dimensional binomial CRR model
to the case where the underlying stock price ratios are distributed over a
bounded interval. The extended binomial model represents a realistic and quite
complex incomplete market model and is of interest to practitioners. At the
same time, practical applications of this extended binomial model are based on
explicit formulas derived in the framework of the basic binomial CRR model.

The present paper builds the foundation for extending the multi-dimensional
binomial model to a more realistic market model, where the underlying stock
price ratios are distributed over a convex set.
\end{remark}

The paper is organized as follows. Section \ref{S:preliminaries} is devoted to introducing
all necessary definitions. We put an emphasis on describing a sample space, an information structure, the sets of martingale and risk-neutral measures for our market model in detail. Here we also provide a survey of known supporting results. We need such a detailed introduction of these basic structures and concepts because they are used both in our numerical algorithms and in proofs of our statements.
In Section \ref{S:single-step} we discuss a single-step model. Here we describe the sets of martingale and risk-neutral measures. We present a problem of linear programming (LP) for computing bounds of the no-arbitrage contingent claim price interval. This section creates building blocks for
the multi-step model discussed in Section~\ref{S:multi-step}. 
Section \ref{S:improvement} is devoted to improvements of the numerical algorithms presented in earlier sections. Here the main Theorem \ref{T:assaf} is formulated and proved. In Section \ref{S:explicit},
we provide explicit pricing formulas for fibrewise supermodular contingent claims. Finally,
in Section \ref{S:examples}, we discuss concrete examples of fibrewise supermodular contingent claims and provide explicit pricing formulas. For European basket call and put options we present explicit formulas for the bounds of no-arbitrage option price intervals in Examples \ref{E:euro-call} and \ref{E:euro-put}. Arithmetic average Asian basket call and put options are discussed in Example \ref{E:asian-call}. 
All technical results needed for the proofs are presented in the Appendix.

\section{Preliminaries}\label{S:preliminaries}

\subsection{The market model and main definitions}\label{SS:model}

We consider a financial market which consists of $m$ risky assets (stocks) and a riskless asset (bond) that can be traded at discrete time moments $t=0,1,\dots,n$. In this market we consider a European contingent claim (option). 

Let us define the main ingredients of the model in more detail.

\subsubsection*{Time} 
The time is modeled by a finite set $\mathbb T=\{0,1,2,\ldots,n\}$.

\subsubsection*{The bond process} It is a deterministic process defined by
$B(t) = R^t$, where $R=1+r$ and $r$ is a constant periodic interest rate.

\subsubsection*{The stock price process} For each $i=1,\ldots,m$, the stock price process $S_i = \left(S_i(t)\right)_{t=0,1,\dots ,n}$ is described by the $n$-step binomial dynamics. The stock price ratios $\psi _i(t)=S_i(t)/S_{i}(t-1)$ at each time moment $t=1,\dots ,n$ can take two possible values: $\psi _i(t)\in \{D_i,U_i\}$. We will assume that for each binomial model the no-arbitrage condition $0<D_i<R<U_i$ holds.

In order to simplify terminology we say that the price of the $i$-th stock
went {\em up} at time $t$ if $\psi _i(t) = U_i$ and it went {\em down}
if $\psi _i(t) = D_i$.
We will denote by $S(t)$ a stock price vector at time $t$:
$$
S(t)=\left(S_1(t),S_2(t),\dots ,S_m(t) \right),
$$
where $t=0,1,\dots,n$.

\subsubsection*{The contingent claim} In the above market, we consider a European contingent claim $X$ with the real-valued pay-off function $f$, which may depend only on the terminal value $S(n)$ of the stock price vector or on the entire stock price path $S(t),$ $t=1,\dots,n$.   

If $m\geq 2$ then the above market model is incomplete \cite[Section 1.5]{pliska}. 

\subsubsection*{Sample space for the $n$-step model and basic random variables} The sample space $\Omega_n$ for the $n$-step model is modeled by the set of $(m\times
n)$-matrices $\omega$ with entries $\omega_{ij}\in~\{0,1\}$. Each element $\omega$ represents the state of the world at time $n$. The value $\omega_{ij}=0$ signifies that the price of the $i$-th stock went down at time $j$. The value $\omega_{ij}=1$ signifies that the price of the $i$-th stock went up at time $j$.

With this notation, the stock price ratio $\psi _i(t)={S_i(t)}/{S_{i}(t-1)}$ for each $t=1, \dots , n$ can be defined as a random variable on $\Omega_n$ described as follows:
$$
\psi_i(t)(\omega) = 
\begin{cases}
D_i & \text{ if } \omega_{it}=0,\\
U_i & \text{ if } \omega_{it}=1,\\
\end{cases}
$$
where $i=1, \dots , m$. Consequently, the $i$-th stock price at time $t\in \BB T$  can be presented as 
$$
S_i(t) = S_i(0) \psi_i(1)\cdots \psi_i(t),
$$
where $i=1, \dots , m$ and $S_i(0)$ is the known initial stock price.

A direct computation shows that
$$
S_i(t)(\omega) = S_i(0) U_i^{\sum_{j=1}^t \omega_{ij}} D_i^{t-\sum_{j=1}^t \omega_{ij}}.
$$

So for each  $\omega\in \Omega_n$ the $i$-th row of $\omega$
describes the $n$-step dynamics of the $i$-th stock price, while the $j$-th column
describes the single-step dynamics of the stock price vector from time $j-1$ to time $j$.

\subsubsection*{Ordering elements in a sample space}
Whenever convenient we shall consider an elementary event $\omega\in \Omega_n$ either as an $(m\times n)$-matrix with entries $\omega_{jk}$, $j=1,\dots , m$, $k=1,\dots ,n$ or as an $n$-tuple of its column vectors
$$
\omega = (\omega^{1}\ \omega^{2}\ \dots\ \omega^{n}),
$$
where $\omega^{k}$ is a column with entries $\omega_{jk}$, $j=1,\dots , m$. It is straightforward to see that $\Omega_n$ consists of $N=2^{mn}$ elements.

It will be
convenient to order the elements of $\Omega_n$ with respect to the reverse
lexicographic order. 
In the case of $n=1$, the set $\Omega_n =\Omega_1$ contains $N=2^{m}$ elements. Each element $\omega\in \Omega_1$ is a $(m\times
1)$-matrix (or, equivalently, a column vector of length $m$), and the elements of $\Omega_1$ are ordered as follows:
\begin{equation}\label{eq:basic_columns}
\omega_1=
\begin{pmatrix}
1 \\
1 \\
\vdots \\
1\\
1
\end{pmatrix},\
\omega_2=
\begin{pmatrix}
1 \\
1 \\
\vdots \\
1\\
0
\end{pmatrix},\ \dots,\
\omega_{N-1}=
\begin{pmatrix}
0 \\
0 \\
\vdots \\
0\\
1
\end{pmatrix},\
\omega_N=
\begin{pmatrix}
0 \\
0 \\
\vdots \\
0\\
0
\end{pmatrix}.
\end{equation}
Each $\omega_i$, when read from top to bottom, is a binary representation of the number $2^m-i.$ Elements of $\Omega_n$ for any $n$ can be ordered in a similar manner.

\subsubsection*{The information structure}

The state of the world at time $k\in \mathbb T$ is described by the subset of
matrices from $\Omega_n$ with the first $k$ columns fixed. Each subset, denoted $\C P(\omega^1,\ldots,\omega^k)$, has the following form: 
\begin{equation}\label{eq:step-k_matrix_coord}
	\C P(\omega^1,\ldots,\omega^k)=
	\left\{	\omega \in \Omega_n\ |\ \omega =\left(\omega^1\ \dots \omega^k\ \ast \dots \ast \right)	\right\}
\end{equation}
In what follows, we will say that the $(m\times k)$-matrix $(\omega^1 \dots \omega^k)$ which represents the common part of all matrices $\omega \in \Omega_n$ included in set $\C P(\omega^1,\ldots,\omega^k)$ is {\em associated with the set $\C P(\omega^1,\ldots,\omega^k)$ }. 

There are $2^{mk}$ disjoint subsets of the form (\ref{eq:step-k_matrix_coord}) with different associated matrices. These subsets form a partition $\C P_k$ of
$\Omega_n$. The initial partition is 
trivial $\C P_0 = \{\Omega_n\}$, the last one
$\C P_n = \{\{\omega_1\},\{\omega_2\},\ldots,\{\omega_N\}\}$, where $N=2^{mn}$.
Clearly, the partition $\C P_k$ is finer than the partition $\C P_{k-1}$
and hence they form a sequence of finer and finer partitions.

Each partition $\C P_k$ can be put into one-to-one correspondence with a subalgebra 
$\C F_k$ of the algebra $2^{\Omega_n}$ of all subsets of $\Omega_n$.
The subalgebras $\C F_k$ form a filtration $\C F$, an increasing sequence of subalgebras $\{\C F_k\},$ $k=0,\dots ,n-1$, where $\C F_k\subseteq \C F_{k+1}.$ Here $\C F_0 =\{\emptyset,\Omega_n\},$ $\C F_n= 2^{\Omega_n}$ consists of all subsets of $\Omega _n.$

In what follows, we assume that the filtration $\C F$ is generated by the stock price vector process $(S(t))_{t=0,1,\dots ,n}.$

\subsubsection*{The supporting tree}\label{SS:supporting-tree} 
The above information structure can be described also with the help of a finite directed rooted $n$-step binary tree which we will denote by $\B T$. We will call $\B T$ the {\em supporting tree} for the $n-$step market model under consideration.
The supporting tree $\B T$ consists of vertices and directed edges that connect the vertices.

In what follows, we will use the following terminology and notation. We will denote by $v_0 \in \B T$ the root of the tree $\B T$. In other words, $v_0$ is the vertex that corresponds to time $t=0.$ For each vertex $v\in \B T$ there is a unique path from the root $v_0$ to $v$. If such a path consists of $k$ edges we
say that the vertex $v$ corresponds to the time step $k$. 
The set of vertices that correspond to time $k$ is denoted by $\B T_k$. In particular, $\B T_n$ are the terminal
vertices. We will call terminal vertices {\em leaves} and non-terminal vertices {\em nodes}. We say that a vertex $v\in \B T_{\ell}$, $\ell \le n$, is a {\em successor} of a node $w\in \B T_k$, $k<\ell$, if there is a path from $w$ to $v$. 
The leaves of the tree $\B T$ are in one-to-one correspondence with the
elements of the sample space $\Omega_n$, or, in other words, with the sets of the partition $\C P_n = \{\{\omega_1\},\{\omega_2\},\ldots,\{\omega_N\}\}$. The set of leaves corresponds to the
time step $n$ and represents all possible states of the world at time $n$.
Additionally, each leaf describes the stock price vector dynamics from time $0$
to time $n$. 

The nodes of $\B T$ that correspond to time  $k<n$ are in one-to-one correspondence with the sets 
$\C P(\omega^1,\ldots,\omega^k)$ of the partition $\C P_k$ that are defined in (\ref{eq:step-k_matrix_coord}). 
In the rest of the paper we will frequently use this correspondence.
So whenever we are talking about a node $v\in \B T_k$ we implicitly
mean that it {\em is} the corresponding set $\C P(\omega^1,\ldots,\omega^k)$ of the partition $\C P_k$. 
The $(m\times k)$-matrix $(\omega^1 \dots \omega^k)$ associated with the set $\C P(\omega^1,\ldots,\omega^k)$ will be called also the matrix {\em  associated with the node $v\in \B T_k$}.
Similarly to the case of leaves, each node $v\in \B T_k$, $k<n$, describes the stock price vector dynamics from time $0$ to time
$k$. 

There is an edge from a node $v$ to a vertex $w$ (where $w$ could be a node or a leaf) if and only if a matrix associated with $w$ has been obtained from the matrix associated $v$
by appending a column. Thus the root is associated with an empty matrix and the tree is regular
in the sense that for each vertex except the leaves and the root there is
exactly one incoming edge and $2^m$ outgoing edges.

\subsubsection*{Probability measures in an $n$-step model}
A probability measure $\B P$ on $(\Omega_n, 2^{\Omega_n}, \C F)$ (or, equivalently, a probability measure in an $n$-step market model) is defined by its probability function $p: \Omega_n \to [0,1]$: 
\[
p(\omega_i)=\B P (\{\omega_i\})=p_i,
\]
where $\omega _i\in \Omega _n$ is an elementary event, $i=1,\dots ,N=2^{mn}$, and $\sum_{i=1}^N p_i = 1$.
For simplicity, in what follows we will use the notation $\B P(\omega_i)$ instead of $\B P (\{\omega_i\})$ and identify a probability measure $\B P$ with a vector $(p_1,p_2,\ldots,p_N)\in \BB R^N$ with non-negative coordinates which
sum up to $1$.  Thus the set of all probability measures in the $n$-step market model with $m$ assets is equivalent to the unit simplex:
\begin{equation}\label{Eq:Delta}
\Delta(\Omega_n) = \left\{(p_1,\ldots,p_N)\in \BB R^N \ |\  \sum_{i=1}^N p_i = 1,\ p_i\geq 0,\ N=2^{mn}\right\}.
\end{equation}
In what follows we will use the notation $\Delta(\Omega_n)$ for the set of probability measures on $(\Omega_n, 2^{\Omega_n}, \C F)$.

\subsubsection*{Martingale and risk-neutral measures.} A probability measure $\B P\in \Delta(\Omega_n)$ is called a {\em martingale} measure in an $n$-step market model if it satisfies the following conditions:
\begin{equation}\label{Eq:martingale_N}
\B E_{\B P}(S_i(k+\ell)|\mathcal{F}_k) = R^{\ell} S_i(k),
\end{equation}
or equivalently, 
\begin{equation}\label{Eq:martingale_N_jump}
\B E_{\B P}(\psi_i(k+1)\dots \psi_i(k+\ell)) = R^{\ell}.
\end{equation}
where $i=1,2,\ldots, m$ and $0\leq k+\ell\leq n$ with $k,\ell\geq 0$. 
In other words, $\B P$ is a martingale measure if and only
if the discounted price process for each stock is a martingale 
with respect to $\B P$. The set of martingale measures in an $n$-step market model will be denoted by ${\rm M}_n$. 

A martingale measure $\B P = (p_1,p_2,\ldots,p_N)$ is called {\em risk-neutral}
if $p_i>~0$ for each $i=1,\ldots,N$. The set of risk-neutral measures in an $n$-step market model will be denoted by ${\rm N}_n$.
So we have: ${\rm N}_n\subset{\rm M}_n\subset \Delta(\Omega_n)$, where ${\rm M}_n$ is the closure of ${\rm N}_n$ in $\BB R^{N}$.

\subsection{Supporting known results}\label{SS:supporting}
In this section we present a number of known results which we need later.

\subsubsection*{Multi-step vs single-step measures}
Let $\B P\in \Delta(\Omega_n)$ be a probability measure. Let us assume in addition $\B P$ is non-degenerate in the following sense: $\B P(\omega_i)=p_i>~0$ for each $i=1,\ldots,N$. For each node $v\in \B T_k$, $k<n$, the
$n$-step non-degenerate probability measure $\B P$ defines a non-degenerate probability measure
$\B P_v\in \Delta(\Omega_1)$ in the corresponding underlying single-step model,
as the following conditional probability given the node $v$:
\begin{equation}\label{eq:single-step-measure}
\B P_v\left(\omega^{k+1}\right) = 
\B P\left(\C P(\omega^1\cdots \omega^k \omega^{k+1}) \ |\ \C P(\omega^1\cdots\omega^k)\right),
\end{equation}
for $k=0,\dots ,n-1$.
Here $\C P(\omega^1\ \ldots\ \omega^k)$ is the set of the partition $\C P_k$
corresponding to the node $v\in \B T_k$, and $\omega^{k+1}$ is an element of
the sample space $\Omega _1$ in the corresponding underlying single-step model
with the root at node $v$. 
Since $\B P$ is non-degenerate, the event $\C P(\omega^1\ \ldots\ \omega^k)$ has
positive measure and the above conditional probability is well defined.

Conversely, assigning a single-step, possibly degenerate, probability 
$\B P_v\in \Delta(\Omega_1)$ at each node $v\in \B T_k$, $k<n$, defines an
$n$-step probability measure $\B P\in \Delta(\Omega_n)$ as follows: for each
$\omega = (\omega^1\ \ldots\ \omega^n)\in\Omega_n$,
\begin{equation}\label{eq:multi-step-measure}
\B P(\omega) 
= \B P_{v_0}(\omega^1)\B P_{v_1}(\omega^2)\cdots \B P_{v_{n-1}}(\omega^n),
\end{equation}
where the node $v_k$ corresponds to the set $\C P(\omega^1 \cdots \omega^k)$.

\subsubsection*{Multi-step and single-step martingale and risk-neutral measures}
The above construction preserves the martingale and risk-neutral properties of measures as stated in the following proposition
the proof of which can be found in \cite[Section 3.4]{pliska}.

\begin{proposition}\label{P:martingale-tree}
Let $\B P\in \Delta(\Omega_n)$ and let $\B T$ be the supporting 
tree. Then $\B P$ is a martingale measure (resp. risk-neutral measure) in an $n$-step market model if and only if
each $\B P_{v_k}\in \Delta(\Omega_1)$ in (\ref{eq:multi-step-measure}) is a martingale measure (resp. risk-neutral measure) in the corresponding underlying single-step model. In other words: 
\begin{align*}
\B P \in {\rm M}_n \; &\Leftrightarrow \; \forall k \; \B P_{v_k}\in {\rm M}_1\\
\B P \in {\rm N}_n \; &\Leftrightarrow \; \forall k \;  \B P_{v_k}\in {\rm N}_1.
\end{align*}
\end{proposition}

\subsubsection*{No-arbitrage pricing of contingent claims}\label{sec: no-arb-price}
Let $X$ be a European type contingent claim in an $n$-step market model.
Since our market model is incomplete for $m\geq 2$, the no-arbitrage price of $X$ at time $k=0,1,\dots ,n-1$ is not unique. Each no-arbitrage price of $X$ at time $k$ is obtained as a discounted conditional expectation with respect to an $n$-step risk-neutral measure $\B P\in {\rm N}_n$ as follows:
\begin{equation}\label{Eq:na-price}
{\rm C}_{\B P}(X,k) =R^{-(n-k)} \B E_{\B P}(X\ |\ \C F_k).
\end{equation}
Notice that for $k>0$, ${\rm C}_{\B P}(X,k)$ is a random variable measurable with respect to 
the algebra $\C F_k$ and hence it is determined by its values on the sets of the
partition $\C P_k$, that is, on the nodes $v\in \B T_k$ of the supporting tree
at time $k$.  By varying the risk-neutral measures we obtain that for each
$v\in \B T_k$ the set of no-arbitrage prices of $X$ is an open interval:
\begin{equation}\label{Eq:price-interval}
({\rm C}_{\min}(v),{\rm C}_{\max}(v))
= \left\{{\rm C}_{\B P}(X,k)(v)\in \BB R \ | \ \B P \in{\rm N}_n \right\}.
\end{equation}

\section{The single-step case}\label{S:single-step}

In this section we consider a single-step model, that is, we assume $n=1$ throughout.

\subsection{Specification of the model}
Recall from Section \ref{S:preliminaries} that the sample space $\Omega_1$ for
the single-step model consists of $(m\times 1)$-matrices with binary
coefficients. Thus $\Omega_1$ has $N=2^m$ elements and they are ordered with
respect to the reverse lexicographical order. 

There are only two instances of time $\BB T=\{0,1\}$ and the
initial prices $S_i(0)$ of the risky assets are known; here $i=1,2,\ldots,m$.
The terminal prices are random variables given by $S_i(1)=S_i(0)\psi_i(1)$, where
$\psi_i(1)\in \{D_i,U_i\}$ are price ratios.

\subsection{Martingale measures}
The system of equations \eqref{Eq:martingale_N_jump} which defines the set ${\rm M}_n$ of martingale measures on $\Omega _n$
takes the following form for $n=1$: 
$$
\B E_{\B P}(\psi_i(1)) = R,
$$
where $i=1,\ldots, m$. It follows that a measure $\B P=(p_1,p_2,\ldots,p_N)$ on $\Omega _1$ is a martingale measure ($\B P\in {\rm M}_1$) if it satisfies the following system of equations and inequalities:
\begin{eqnarray}\label{eq:mart_meas_n=1}
\psi_1(\omega_1)p_1+ \psi_1(\omega_2)p_2 + \cdots + \psi_1(\omega_N)p_N&=&R \\
\cdots\qquad  &\cdots &\nonumber \\
\psi_m(\omega_1)p_1+ \psi_m(\omega_2)p_2 + \cdots + \psi_m(\omega_N)p_N&=&R\nonumber \\
p_1+p_2+\cdots+p_N&=&1\nonumber\\
p_j &\geq 0&.\nonumber
\end{eqnarray}
Here we used the simplified notation $\psi_i(\omega_j)$ to denote $\psi_i(1)(\omega_j)$.

Let ${\B \Psi}$ be an $(m\times N)$-matrix that corresponds to the 
first $m$ equations
in the above system. That is,
\begin{equation}\label{eq:Psi}
\B \Psi_{ij} = \psi_i(\omega_j)
\end{equation}
for $i=1,\ldots,m$ and $j=1,\ldots, N=2^m$. 

It follows from the above discussion that ${\rm M}_1$ (the set of martingale measures on $\Omega _1$) is a subset of $\Delta(\Omega_1)$ (the set of all probability measures on $\Omega_1$) such that each $\B P\in {\rm M}_1$ satisfies the system of linear equations 
\begin{equation}\label{eq:mart_cond_n=1}
\B \Psi\ \B P = \B R,
\end{equation}
where $\B \Psi$ is given by \eqref{eq:Psi} and $\B R$ is a column vector of size $m$ with all entries equal to $R$. Let $\B A\subset \BB R^N$ be an affine subspace of solutions of \eqref{eq:mart_cond_n=1}:
$$
\B A = \left\{\B P\in \BB R^N\ | \ \B \Psi\ \B P = \B R\right\}.
$$
We can now summarize the description of the set ${\rm M}_1$ of martingale measures on $\Omega_1$. Geometrically ${\rm M}_1$ forms a subset of $\BB R^N$ which is an intersection of the simplex $\Delta(\Omega_1) \subset \BB R^N$ of all probability
measures on $\Omega_1$ with an affine subspace $\B A$: 
$$
{\rm M}_1 = \Delta(\Omega_1) \cap \B A.
$$
Thus ${\rm M}_1$ is a bounded
convex polytope or, in other words, the convex hull of finitely many points
called vertices. Recall that the risk-neutral measures are those martingale
measures $\B P=(p_1,\ldots,p_N)$ for which $p_j>0$ for all $j=1,\ldots, N$. 
Thus the polytope ${\rm M}_1$ of martingale measures on $\Omega_1$ is the closure of ${\rm N}_1$, the set
of risk-neutral measures on $\Omega_1$.
\begin{remark}
An example at the beginning of Section \ref{SS:nonexamples}  provides an explicit solution to \eqref{eq:mart_meas_n=1} for the case of two assets.
\end{remark}

\subsection{Interval of no-arbitrage continent claim prices}
Consider a European contingent claim $X = f(S(1))$ in the single-step
model. Here 
$$
S(1) = (S_1(1),S_2(1),\ldots,S_m(1))
$$ 
is the stock price vector at maturity and
$f\colon \BB R^m\to \BB R$ is the pay-off function of $X$.

In this single-step model, there is only one node on the supporting tree $\B T$. It is a root of $\B T$ that corresponds to time $t=0$. Therefore, there is only one open interval of the no-arbitrage prices of $X$ (see \eqref{Eq:price-interval}). We will use a simplified notation $({\rm C}_{\min}(0),{\rm C}_{\max}(0))$ for that open interval.

Recall (see Section \ref{sec: no-arb-price}) that each no-arbitrage price of $X$ at time zero ${\rm C}_{\B P}(X,0)$ (or, equivalently, each point in the above open interval) is obtained by computing the expectation of $X$ with respect to a risk-neutral measure on $\Omega _1$ discounted to time zero:
\begin{equation}\label{Eq:na-price-zero}
{\rm C}_{\B P}(X,0) =R^{-1} \B E_{\B P}(X),
\end{equation}
where $\B P\in {\rm N}_1$. Since the set ${\rm M}_1$ of martingale measures on $\Omega _1$ is the closure of the
set ${\rm N}_1$ of risk-neutral measures on $\Omega _1$ we obtain that  
$$
[{\rm C}_{\min}(0),{\rm C}_{\max}(0)] = 
\left\{{\rm C}_{\B P}(X,0)\in \BB R\ |\ \B P \in {\rm M}_1\right\}.
$$
It follows that the upper and lower bounds of the above interval are:
\begin{eqnarray}
{\rm C}_{\max}(0)&=&\max\{{\rm C}_{\B P}(X,0)\ |\ \B P\in {\rm M}_1 \}\label{eq:ub}\\
{\rm C}_{\min}(0)&=&\min\{{\rm C}_{\B P}(X,0)\ |\ \B P\in {\rm M}_1 \}. \label{eq:lb}
\end{eqnarray}
Finding the quantities ${\rm C}_{\max}(0)$ and ${\rm C}_{\min}(0)$ is an important problem of contingent claim pricing.

\subsection{Computing bounds of the no-arbitrage contingent claim price interval}\label{SS:simplex}
With the given contingent claim $X$ we will associate a vector 
$$
\B X=(X_1,X_2, \ldots ,X_N)\in\BB R^N,
$$ 
where
\begin{equation}\label{eq:vectorX}
X_j=f(S_1(1)(\omega_j),\ldots,S_m(1)(\omega_j)),
\end{equation}
where $j=1,\ldots,N$. Then the expected value of $X$ with respect to a probability measure $\B P \in \Delta(\Omega_1)$ can be presented as follows:
\begin{equation}\label{eq:exp_value}
\B E_{\B P}(X)=\sum_{i=1}^{N}X_i p_i=\langle \B X,\B P\rangle,
\end{equation}
where $\langle \cdot, \cdot \rangle$ denotes the Euclidean scalar product in $\BB {R}^N$. 

For a fixed $X$, a map ${\rm C}_{\cdot}(X,0)$ is a linear functional on $\BB R^N$:
\begin{equation}\label{Eq:lin-funct}
{\rm C}_{\B P}(X,0)= R^{-1}\langle \B X,\B P\rangle.
\end{equation}
So the closed interval $[{\rm C}_{\min}(0),{\rm C}_{\max}(0)]$ is the
image of the set ${\rm M}_1$ of martingale measures on $\Omega_1$ with respect
to this linear functional.

As a result, the problem of finding bounds \eqref{eq:ub}-\eqref{eq:lb} for the
no-arbitrage price interval of a contingent claim $X$ can be formulated in
terms of the problem of finding extrema for a linear functional ${\rm
C}_{\cdot}(X,0)$ on a convex polytope ${\rm M}_1$ in $\BB {R}^N$.

As is well known, these extrema are attained at the vertices of ${\rm M}_1$.
Moreover finding both the extremal values of ${\rm C}_{\cdot}(X,0)$ and the
vertices of ${\rm M}_1$ at which they are attained is the standard problem of
Linear Programming (LP).  See, for example, \cite{MR1851303} for various
examples of suitable algorithms.

The following describes the algorithm to compute the upper bound ${\rm
C}_{\max}(0)$ of the interval of no-arbitrage prices for a given contingent
claim $X$.

{\bf Given data:}
\begin{itemize}
\item
Initial stock prices: $S_1(0),\dots,S_m(0)$.
\item
Risk-free growth factor: $R>0$.
\item
Parameters of the binomial models: $\{D_i,U_i\}$, where $i=1, \dots , m.$
\item 
Pay-off function $f.$
\end{itemize}
{\bf Step one: compute input data}
\begin{itemize}
\item 
Compute entries $\B \Psi_{ij}$ of the matrix $\B \Psi$ for $i=1,\dots,m$ and $j=1,\dots, N=2^m$ (see \eqref{eq:Psi} ).
\item
Compute terminal stock prices: $S_i(1)(\omega_j)=S_i(0)\B \Psi_{ij} $, where $i=1,\dots, m$ and $j=1,\dots, N$.
\item
Compute the pay-off at maturity: $X_j=f(S_1(1)(\omega_j),\dots,S_m(1)(\omega_j))$, $j=1,\dots, N$.
\end{itemize}
{\bf Step two: apply a suitable linear programming algorithm}

Solve the following LP problem.

Maximize function:
\begin{equation}\label{eq:LP}
\B P\mapsto {\rm C}_{\B P}(X,0)=R^{-1}\langle \B X,\B P\rangle=R^{-1}\left( X_1 p_1 + \cdots + X_{N} p_{N}\right),
\end{equation}
subject to constraints:
\begin{align*}
\B \Psi_{11}p_1 + \cdots + \B \Psi_{1N}p_{N} &= R\\
&\ \vdots\\
\B \Psi_{m1}p_1 + \cdots + \B \Psi_{mN}p_{N} &= R\\
p_1 + \cdots + p_{N} &= 1\\
p_i &\geq 0
\end{align*}
{\bf Output:} 
\begin{itemize}
	\item The upper  bound ${\rm C}_{\max}(0)$ of the no-arbitrage price interval for $X$ at $t=0$.
	\item The maximal martingale measure $\B P_{\max}$, or, in other words, the vertex $\B P_{\max} = (p^*_1,\ldots,p^*_{N})\in {\rm M}_1$ at which
	the maximum ${\rm C}_{\max}(0)$ is attained. 
\end{itemize}
Notice that the maximum of a functional can be attained at several vertices and
the algorithm will choose one of them in such a case.

The lower bound ${\rm C}_{\min}(0)$ of the no-arbitrage price interval for $X$ at $t=0$ can be found along similar lines with appropriate modifications in the algorithm.

\begin{remark}
	In \cite[Section 3]{MR2283328}, the LP problem \eqref{eq:LP} is solved analytically for a special case of two assets. Explicit formulas for the bounds of the no-arbitrage option price interval as well as for the corresponding extremal martingale measures are obtained. 
\end{remark}

\section{The multi-step case}\label{S:multi-step}

In this section we consider the $n$-step market model with $m$ risky assets
and a contingent claim $X$ with the pay-off function $f$.

We are going to extend the results of the previous section to this multi-step case. Let us first recall a few important facts from Section \ref{S:preliminaries} and introduce some convenient notation.

\subsection{Quick review and necessary notation}\label{SS:quick-review}
A sample space of the $n$-step model is $\Omega_n$ and there is an
associated supporting tree $\B T$ in which the set of vertices at time $k\in \BB T$
is denoted by $\B T_k$. We say that the set $\B T_n$ consists of leaves, and each set $\B T_k$, $k<n$ consists of nodes. Each set $\B T_k$ , $k=0,1,\dots ,n$ is in a bijective correspondence with the
partition $\C P_k$ of the sample space $\Omega_n$.

Any probability measure $\B P\in \Delta(\Omega_n)$ in the $n$-step model defines a set of single-step probability measures $\B P_v\in \Delta(\Omega_1)$ for each node $v$ of~$\B T$. 

Let $\omega = (\omega^1,\ldots,\omega^n)\in \Omega_n$ be an element of the sample space. It defines a unique path on the tree $\B T$ from the root
to the leaf it represents. Such a path is a sequence of nodes which we denote
by
$$
v_0(\omega),v_1(\omega),\ldots,v_n(\omega),
$$
where $v_0(\omega)$ is the root. Then for any measure $\B P\in \Delta(\Omega_n)$ the value $\B P(\omega)$ can be presented as follows: 
\begin{equation}\label{Eq:P(w)}
\B P(\omega) = \B P_{v_0(\omega)}\left(\omega^1\right)
\B P_{v_1(\omega)}\left(\omega^2\right)\cdots
\B P_{v_{n-1}(\omega)}\left(\omega^n\right).
\end{equation}
We will use the above notation to write down a conditional expectation of a random variable $X$ on $\Omega_n$ with respect to a measure $\B P\in \Delta(\Omega_n)$. Let $\C F_k$ be a subalgebra of the algebra $2^{\Omega_n}$ of all subsets of $\Omega_n$. The conditional expectation $\B E_{\B P}(X|\C F_k)$ is a random variable measurable with respect to the algebra $\C F_k$. Therefore it is determined by its values $\B E_{\B P}(X|\C F_k)(v)$ on the nodes $v\in \B T_k$. 

Let us fix a node $v\in \B T_k$ and let $\C P(v)$ denote the set of the partition $\C P_k$ which corresponds to the node $v$. It follows straightforwardly from the properties of the conditional expectation that the value $\B E_{\B P}(X|\C F_k)(v)$ can be computed recursively starting from the leaves of the supporting tree $\B T$ and, speaking informally, folding back the tree, by computing expectations with respect to suitable single-step measures:
\begin{equation}\label{Eq:E(X|F)}
	\B E_{\B P}(X|\C F_k)(v) = \sum_{\omega\in \C P(v)}X(\omega)\frac{\B P(\omega)}{\B P(\C P(v))} = \sum_{\omega\in \C P(v)} X(\omega) \prod_{j=k+1}^n \B P_{v_{j-1}(\omega)}(\omega^j).
\end{equation}
Here $\B P_{v_{j}(\omega)}$ are the single-step measures that occur in \eqref{Eq:P(w)}. Recall that if $\B P$ is a risk-neutral (resp. martingale) measure on $\Omega_n$, then $\B P_{v}$ is a risk-neutral (resp. martingale) measure on $\Omega_1$ for each node $v$ of~$\B T$, and vise versa. 

\subsection{Extremal martingale measures}\label{SS:extremal-measures}
Let $X$ be a contingent claim. Given an $n-$step risk-neutral measure $\B P\in {\rm N}_n$, a no-arbitrage price of $X$ at time $k$ with respect to $\B P$ can be computed as the conditional expectation of $X$ with respect to an algebra $\C F_k$ discounted to time $k$:
\begin{equation}\label{Eq:na-price-1}
{\rm C}_{\B P}(X,k) =R^{-(n-k)} \B E_{\B P}(X|\C F_k).
\end{equation}
The time $k$ no-arbitrage price of $X$ is a random variable measurable with respect to the algebra $\C F_k$. We will use the notation 
${\rm C}_{\B P}(X,k)(v)$ for the price of $X$ which corresponds to a node $v\in \B T_k$.
Using the representation \eqref{Eq:E(X|F)} we have:
\begin{equation}
{\rm C}_{\B P}(X,k)(v) =R^{-(n-k)}\sum_{\omega\in \C P(v)} X(\omega) \prod_{j=k+1}^n \B P_{v_{j-1}(\omega)}(\omega^j).
\label{Eq:no-arb-price-k}
\end{equation}
Varying the risk-neutral measures $\B P\in {\rm N}_n$, we obtain at each node $v\in \B T_k$ an open interval of no-arbitrage prices for a given contingent claim $X$ (see \eqref{Eq:price-interval}):
$$
({\rm C}_{\min}(v),{\rm C}_{\max}(v))
= \left\{{\rm C}_{\B P}(X,k)(v)\in \BB R \ | \ \B P \in{\rm N}_n \right\}.
$$
The upper and lower bounds of this interval are:
\begin{eqnarray}
{\rm C}_{\min}(v)&=&\inf\{{\rm C}_{\B P}(X,k)(v)\ |\ \B P\in {\rm M}_n \} \label{eq:lb-n}\\
{\rm C}_{\max}(v)&=&\sup\{{\rm C}_{\B P}(X,k)(v)\ |\ \B P\in {\rm M}_n \}.\label{eq:ub-n}
\end{eqnarray}
Recall that ${\rm M}_n$ stands for the set of $n-$step martingale measures on $\Omega_n$ and ${\rm M}_n$ is a closure of ${\rm N}_n$. 

Our goal is to identify the bounds ${\rm C}_{\min}(v)$ and ${\rm C}_{\max}(v)$ for each node $v$ on the supporting tree $\B T$. We will start with the discussion of extremal martingale measures that produce these bounds. 

Let $0\in \B T$ be the root and let $\B P_{\min}\in {\rm M}_n$ and $\B P_{\max}\in {\rm M}_n$ be the extremal martingale measures
for which the bounds ${\rm C}_{\min}(0)$ and ${\rm C}_{\max}(0)$ (respectively) are attained. In other words,
\begin{eqnarray}
{\rm C}_{\min}(0)&=&R^{-n} \B E_{\B P_{\min}}(X) \label{eq:lb-n-0}\\
{\rm C}_{\max}(0)&=&R^{-n} \B E_{\B P_{\max}}(X).\label{eq:ub-n-0}
\end{eqnarray}
The following proposition shows that the extremal measures
$\B P_{\min}$ and $\B P_{\max}$ produce the bounds of no-arbitrage price intervals for $X$ also at
all other nodes of the tree.

\begin{proposition}\label{P:Pmax}
Let $\B P_{\min}\in {\rm M}_n$ and $\B P_{\max}\in {\rm M}_n$ be the extremal martingale measures defined in \eqref{eq:lb-n-0} and \eqref{eq:ub-n-0} respectively for a given contingent claim $X$. 

Then for any $k\in\{0,1,\ldots,n-1\}$ and any node $v\in \B T_k$ the following equalities
are true:
\begin{eqnarray}
{\rm C}_{\min}(v)&=&R^{-(n-k)}\B E_{\B P_{\min}}(X|\C F_k)(v) \label{eq:lb-n-v}\\
{\rm C}_{\max}(v)&=&R^{-(n-k)}\B E_{\B P_{\max}}(X|\C F_k)(v),\label{eq:ub-n-v}
\end{eqnarray}
where the bounds ${\rm C}_{\min}(v)$ and ${\rm C}_{\max}(v)$ are defined in \eqref{eq:lb-n} and \eqref{eq:ub-n}, respectively.
\end{proposition}

\begin{proof}
We will present the proof for $\B P_{\max}$ leaving an analogous proof
for $\B P_{\min}$ to the reader.

Denote the measure $\B P_{\max}$ by $\B P$ for simplicity. Suppose that
there exists a different martingale measure $\B P'\in {\rm M}_n$ $(\B P'\ne \B P)$ and a node $v\in \B T_k$ such that the upper bound
${\rm C}_{\max}(v)$ of the no-arbitrage price interval for $X$ at $v$ is attained for $\B P'$, rather than for $\B P$. In other words, suppose that the following holds:
$$
\B E_{\B P}(X|\C F_k)(v) < \B E_{\B P'}(X|\C F_k)(v).
$$
According to formula \eqref{Eq:E(X|F)}, this
is equivalent to 
\begin{align}
\sum_{\omega\in \C P(v)} X(\omega) \prod_{j=k+1}^n \B P_{v_{j-1}(\omega)}(\omega^j)
&< \sum_{\omega\in \C P(v)} X(\omega) \prod_{j=k+1}^n \B P'_{v_{j-1}(\omega)}(\omega^j).
\label{Eq:P<P'}
\end{align}
Let us now compute the expected value $\B E_{\B P}(X)$. According to our assumption \eqref{Eq:P<P'}, we have that
\begin{align*}
	\B E_{\B P}(X)&= 
	\sum_{\omega\in \Omega_n}X(\omega)\B P(\omega)=\sum_{\omega\in \C P(v)}X(\omega)\B P(\omega) + 
	\sum_{\omega\notin \C P(v)}X(\omega)\B P(\omega)\\ 
	&=\sum_{\omega\in \C P(v)} X(\omega) \prod_{j=1}^n \B P_{v_{j-1}(\omega)}(\omega^j)
	+\sum_{\omega\notin \C P(v)}X(\omega)\B P(\omega)\\ 
	&=\prod_{j=1}^k \B P_{v_{j-1}(\omega)}(\omega^j)\sum_{\omega\in \C P(v)} X(\omega) \prod_{j=k+1}^n \B P_{v_{j-1}(\omega)}(\omega^j)
	+\sum_{\omega\notin \C P(v)}X(\omega)\B P(\omega)\\ 
	&<\prod_{j=1}^k \B P_{v_{j-1}(\omega)}(\omega^j)\sum_{\omega\in \C P(v)} X(\omega) \prod_{j=k+1}^n \B P'_{v_{j-1}(\omega)}(\omega^j)
	+\sum_{\omega\notin \C P(v)}X(\omega)\B P(\omega)\\
	&=\B E_{\B P''}(X), 
\end{align*}
where $\B P''\in {\rm M}_n$ is a martingale measure obtained by replacing each single-step measure $\B P_{v}$ (which corresponds to the $n-$step measure $\B P$) by the single-step measure $\B P'_{v}$ (which corresponds to the $n-$step measure $\B P'$) at the node $v$ and further at all the successors of $v$ on the tree. It follows from the above estimate that there exists a martingale measure $\B P''$ for which 
$$
R^{-n}\B E_{\B P''}(X) > {\rm C}_{\max}(0),
$$
which contradicts the assumption that ${\rm C}_{\max}(0)$ is the upper bound of the no-arbitrage price interval for $X$ at zero. The argument for $\B P_{\min}$ is analogous.
\end{proof}

\subsection{Computing bounds of the no-arbitrage contingent claim price interval.}\label{SS:bounds}

It is a consequence of the above proposition that the bounds 
${\rm C}_{\min}(0)$ and ${\rm C}_{\max}(0)$ of the no-arbitrage price interval
for a given contingent claim $X$ can be computed recursively starting from the leaves of the supporting tree $\B T$ and going backwards in time. 

More specifically. For each penultimate node $v\in \B T_{n-1}$ we solve a single-step LP problem described in Section \ref{SS:simplex}, where $v$ plays the role of a root of the corresponding single-step tree. The option payoff values are computed at the leaves adjacent to $v$ and are used as input data for a linear programming algorithm described in Section \ref{SS:simplex}. Solving the single-step LP problem, we find a measure $\B P_{\max,v}$, a maximal martingale measure for this single-step problem and determine the upper bound ${\rm C}_{\max}(v)$ at the node $v$. 

Once the upper bound of the no-arbitrage option price interval at
each node of $\B T_{n-1}$ is computed, the same procedure is applied at each
node of $\B T_{n-2}$. Consider a node $w\in\B T_{n-2}$. This node will play the role of a root of the corresponding single-step tree. Each node $v\in \B T_{n-1}$ adjacent to $w$ will play the role of the leaf of a single-step tree rooted at $w$. Each upper bound ${\rm C}_{\max}(v)$ at node $v$ will replace the corresponding option pay-off value at that node. 

Continuing this way, we arrive at the root of the $n-$step tree and determine ${\rm C}_{\max}(0)$.
The argument for the lower bound ${\rm C}_{\min}(0)$ is analogous. 

This algorithm works for a general European type contingent claim. It is, however, infeasible from a
computational point of view for the multi-step models, because the number of nodes
grows exponentially with the number of time steps in the model. We discuss the improvement
of the algorithm in Section \ref{S:improvement}, where we show that for 
contingent claims with pay-off functions from a special class one can significantly reduce the computational complexity.

\section{Improvements of the algorithm}\label{S:improvement}

\subsection{The recombinant graph}\label{SS:recombinant}

The recursive algorithm for a computation of the bounds of the no-arbitrage contingent claim interval described in the previous section can be improved by descending
to the recombinant graph, which we define as follows.

Let $\B T$ be the supporting tree for our $n$-step market model.
Introduce the following equivalence relation on the vertices of $\B T$. 
\begin{definition}\label{D:equivalence}
Let $u,v\in\B T_k$, $0<k\le n$. Let $\omega _u$ and $\omega _v$ be $(m\times k)$-matrices associated with $u$ and $v$ respectively. Vertices $u$ and $v$ are called {\em equivalent} if and only if the sum of row entries of the matrix $\omega _u$ is equal to the sum of row entries of the matrix $\omega _v$ for each row. For a given vertex $v$, we will denote by $[v]$ an equivalence class consisting of vertices that are equivalent to $v$.
\end{definition}

\begin{definition}\label{D:recombinant} 
Let $\B T_{\rm r}$ be a directed rooted graph with the vertex set consisting of
the above equivalence classes and such that there is an edge from $[v]$ to
$[u]$ if and only if there is an edge from $v$ to $u$ in the tree $\B T$.  The
graph $\B T_{\rm r}$ is called the {\em recombinant graph}.  
\end{definition}

Similar to the case of the supporting tree $\B T$, we will call terminal
vertices of the recombinant graph $\B T_{\rm r}$ leaves, and non-terminal
vertices of $\B T_{\rm r}$ nodes.  Observe that each node of the recombinant
graph has $2^m$ outgoing edges, however, the number of incoming edges can vary.
Thus each node and its descendants form the tree which corresponds to the
appropriate single-step model.  The main advantage of a recombinant graph is
that the number of its vertices grows polynomially in time while in the
corresponding tree it grows exponentially.

The recursive algorithm for computing the bounds of the no-arbitrage contingent claim intervals is essentially the same as described before and it works for a general
contingent claim of a European type.  For each penultimate node we
solve the single-step optimization problem and proceed recursively to the root.  Since the
recombinant graph in the $n$-step model with $m$ assets has $(k+1)^m$ nodes at
level $k\leq n$, running the algorithm for small $m$ and not too big $n$
becomes feasible.

\begin{example}\label{E:program-recombinant}
We have tested the above improvement by running a computer program that computes
the vertices of the polytope of martingale measure, finds
the maximal martingale measure and the values of $C_{\max}(v)$ for each node of the recombinant graph. For
example, the program terminates within a few seconds for $m=5$ assets and
$n\leq 8$. Also it takes a second to compute the above data for
$m=12$ assets and $n=1$ step. This is relevant in view of the next improvement.
\end{example}

\subsection{Bounds of the no-arbitrage price interval for special contingent claims}\label{SS:improvement}
For contingent claims that belong to a certain class (fibrewise supermodular contingent claims), further improvements of the algorithm described in Section \ref{SS:bounds} are available. For such contingent claims, both the maximal martingale measure and the minimal martingale measure are product measures. As a result, the bounds of the no-arbitrage contingent claim price interval can be computed by means of solving the LP problem described in Section \ref{SS:simplex} only once.
\begin{theorem}\label{T:assaf}
Suppose a contingent claim $X\colon \Omega_n\to \BB R$
is fibrewise supermodular (see Definition \ref{D:f-supermodular}).
\begin{enumerate}[label=(\roman*)]
	\item 
	\label{T:cont_claim:max}
There exists a single-step martingale measure $\B P\in {\rm M}_1$ such that
$\B P_{\!\!\max} = \B P\otimes \cdots \otimes \B P = \B P^n$.
That is, the maximal martingale measure is a product measure.
	\item
	\label{T:cont_claim:min}
	If $m=2$ or $\sum_{i=1}^m \frac{R-D_i}{U_i-D_i}\leq 1$ then the minimal
	martingale measure is also product:
	$\B P_{\min}=(\B P')^{n}$, for some single-step 
	martingale measure $\B P'$.
\end{enumerate}
\end{theorem}
\begin{proof}
\ref{T:cont_claim:max}. This part follows immediately from Theorem \ref{T:min submodular n step}, \ref{T:fibrewise_modular:max} (see Appendix), with $b_i=\frac{R-D_i}{U_i-D_i}$, $i=1,\dots ,m$ (see \eqref{E:b_i}).

\ref{T:cont_claim:min}. If $m>2$, the statement follows immediately from Theorem \ref{T:min submodular n step}, \ref{T:fibrewise_modular:min} (see Appendix). 
In the case of $m=2$ assets, the single step martingale measure form an interval and the extremal measures are its endpoints. If the maximal martingale measure
in an $n$-step model is a product $\B P_{\max}=\B P$ then it means that
$\B P$ is an endpoint of the above interval. Thus the minima are attained
on the other endpoint $\B P'$. It follows that $\B P_{\min} = (\B P')^n$.
\end{proof}
\begin{corollary}\label{C:upper_bound}
	Under the conditions of Theorem \ref{T:assaf}, for a fibrewise supermodular contingent claim $X$, one has the following:
\begin{equation}\label{Eq:cmax}
	{\rm C}_{\max}(v) =	\B E_{\B P^n}(X|\C F_k)(v) = \sum_{\omega\in \C P(\omega^1\cdots \omega^k)}\ \prod_{j=k+1}^n 
	\B P(\omega^{j})X(\omega),
\end{equation}
for $v\in \B T_k$ corresponding to the set 
	$\C P(\omega^1\cdots\omega^k)\in \C P_k$.  The summation goes over the elements
	$\omega = (\omega^1\cdots \omega^n)\in \Omega_n$ with fixed first $k$ columns.
	Moreover,
	\begin{equation}\label{Eq:cmin}
			{\rm C}_{\min}(v)=	\B E_{(\B P')^n}(X|\C F_k)(v)
		= \sum_{\omega\in \C P'(\omega^1\cdots \omega^k)}\ \prod_{j=k+1}^n 
		\B P'(\omega^{j})X(\omega),
		\end{equation}
\end{corollary}
In what follows we describe the maximal and the minimal martingale measure explicitly and
evaluate the above formulas.

\section{Pricing fibrewise supermodular contingent claims}\label{S:explicit}
\subsection{Explicit formula for the upper bound of a no-arbitrage price interval}\label{SS:max-measure}
The one-step maximal martingale measure $\B P$ from Theorem \ref{T:assaf} can be
described explicitly as follows (see Equation \eqref{Eq:supervertex}).
Let $\mu_k = (1,1,\ldots,1,0,\ldots,0)^{\sf T}$ be the column vector with
the first $k$ entries equal to $1$ and the rest zero. 
Then
\begin{equation}
\B P(\omega)=
\begin{cases}
b_i-b_{i+1} & \text{ if } \omega=\mu_i\\
0 & \text{ otherwise},
\end{cases}
\label{Eg:supervertex}
\end{equation}
where $b_i = \frac{R-D_i}{U_i-D_i}$ and $b_0=1$. Notice that this measure is
highly degenerate because among its $2^m$ entries only $m$ of them are nonzero.

\begin{example}\label{E:ns}
If $m=2$ then
$$
\B P\left( 
\begin{smallmatrix}
1\\
1
\end{smallmatrix}
 \right)
=b_2,
\
\B P\left( 
\begin{smallmatrix}
1\\
0
\end{smallmatrix}
 \right)
=0,\
\B P\left( 
\begin{smallmatrix}
0\\
1
\end{smallmatrix}
 \right)
=b_1-b_2,\
\B P\left( 
\begin{smallmatrix}
0\\
0
\end{smallmatrix}
 \right)
=1-b_1
$$
which agrees with the formula for $Q_{\lambda_+}$ in \cite[Remark 1]{MR2283328}.
\hfill $\diamondsuit$
\end{example}

Thus in the formula \eqref{Eq:cmax} the summation takes place over matrices
with fixed $k$ columns and the remaining columns are one of the $\mu_i$'s only:
\begin{equation}
C_{\max}(v) =
\sum_{i\in I^{n-k}} \B P(\mu_{i_1})\cdots\B P(\mu_{i_{n-k}}) 
X(\omega^1\cdots\omega^k\mu_{i_1}\cdots\mu_{i_{n-k}}),
\label{Eq:cmax-better}
\end{equation}
where $i=(i_1,\ldots,i_{n-k})$, $i_j\in I=\{0,1,\ldots,m\}$ and the first $k$
columns $\omega^1\cdots\omega^k$ correspond to the vertex $v\in \B T_k$. Notice
that even in this case the number of summands for $C_{\max}(0)$ is exponential
in $n$. However, if the payoff is path-independent then many of these summands
are equal and we get:
\begin{align}
\label{Eq:cmax-path-ind}
&C_{\max}(v)=\\
&\sum_{k_0+\ldots+k_m=n} \frac{n!}{k_0!\cdots k_m!} \B P(\mu_0)\cdots\B P(\mu_m) 
X\left( \omega^1\cdots\omega^k
\underbrace{\mu_0\cdots\mu_0}_{k_0 \text{ times}}\cdots
\underbrace{\mu_m\cdots\mu_m}_{k_m \text{ times}} \right)
\nonumber
\end{align}

\subsection{Explicit formula for the lower bound of a no-arbitrage price interval}\label{SS:minimal}
Let $\nu_i \in \BB R^m$, for $i=1,\ldots,m$ be the standard basis 
vector. That is, the $i$-th coordinate of $\nu_i$ is equal to $1$ and the others are zero. 
Suppose that $\sum_{i=1}^m b_i\leq 1$, where $b_i=\frac{R-D_i}{U_i-D_i}$, $i=1,\dots ,m$ that is, the assumption of
Theorem \ref{T:assaf}, \ref{T:cont_claim:min} is satisfied.
It follows from definition given in Equation \eqref{Eq:subvertex} in the Appendix that
the minimal one-step martingale measure $\B P'$ is then given by
$$
\B P'(\omega)=
\begin{cases}
1-\sum_{i=1}^m b_i & \text{ if } \omega=(0,\ldots,0)^{\sf T}\\
b_i & \text{ if } \omega=\nu_i\\
0 & \text{ otherwise}.
\end{cases}
$$ 

\begin{example}\label{E:ns-min}
If $m=2$ we have
$$
\B P'
\left(  
\begin{smallmatrix}
1\\
1
\end{smallmatrix}
\right)
=0,\
\B P'
\left(  
\begin{smallmatrix}
1\\
0
\end{smallmatrix}
\right)
=b_1,\
\B P'
\left(  
\begin{smallmatrix}
0\\
1
\end{smallmatrix}
\right)
=b_2,\
\B P'
\left(  
\begin{smallmatrix}
0\\
0
\end{smallmatrix}
\right)
=1-b_1-b_2,
$$
which agrees with the formula for $Q_{\lambda_-}$ in \cite[Remark 1]{MR2283328}.
\hfill $\diamondsuit$
\end{example}

The formulas for $C_{\min}(v)$ are analogous to the formulas \eqref{Eq:cmax-better} and
\eqref{Eq:cmax-path-ind} for $C_{\max}(v)$ in which the factors 
$\B P(\mu_i)$ are replaced by $\B P'(\nu_i)$.

\section{Applications of Theorem \ref{T:assaf} and concrete examples}
\label{S:examples}
\subsection{Fibrewise supermodular contingent claims}\label{SS:fs-payoff}
In the present setting
a random variable $X\colon \Omega_n\to~\BB R$ is fibrewise supermodular 
(cf. Definition \ref{D:f-supermodular}) if
for each $k=1,2,\ldots, n$ its restriction to the subset consisting of all
entries, except those in the $k$-th column, fixed is supermodular (see 
definition on page \pageref{SS:supermodular}).

In what follows we present a fairly general construction which we will
subsequently specify to concrete examples of contingent claims. 

\begin{definition}\label{D:R_polynomial}
	A function
	$p\colon \BB R^k\to \BB R$ is called $\BB R$-polynomial in $k$ variables if
	$p(x_1,\ldots,x_k)$ is a linear combination of the Cobb-Douglas
	functions: $x_1^{p_1}\cdots x_k^{p_k}$, where $0\leq p_i\in \BB R$; see
	\cite[Proposition 2.2.4]{MR3154633}. In an ordinary polynomial the exponents
	$p_i$ are non-negative integers.
\end{definition}

\begin{lemma}\label{L:f-supermodular}
Let $h\colon \BB R\to \BB R$ be a convex function and let $p(x_{11},\ldots,x_{mn})$ be
an $\BB R$-polynomial in $mn$ variables with non-negative coefficients. Let $S_i(j)$, $i=1,\dots ,m$, $j=1, \dots ,n$ be stock price values in the $n$-step market model with $m$ assets.
Then the random variable $X:\Omega_n\to \BB R$
$$
X = h\left( p(S_1(1),\ldots,S_m(1),\ldots,S_1(n),\ldots,S_m(n))\right)
$$
is fibrewise supermodular.
\end{lemma}
\begin{proof}
Since $S_i(k) = S_i(0)\psi_i(1)\cdots\psi_i(k)$, the polynomial in the stock prices $S_i(j)$
is also a polynomial in the stock price ratios $\psi_i(j)$. By restricting it to the element
of the sample space $\Omega_n$ consisting of matrices with fixed all columns but the
$k$-th one we obtain an $\BB R$-polynomial (with non-negative coefficients) in $\psi_i(k)$.
The statement then follows from Proposition 2.2.4 (b) and Proposition 2.2.5 (a) in
\cite{MR3154633}.
\end{proof}

\begin{example}[European basket call option]\label{E:euro-call}
Let $K\geq0$ and let $h(x) = (x-K)^+ :=\max\{x-K,0\}$ and consider a random variable
$$
X = h\left( \sum_i a_iS_i(n)\right)=\left( \sum_i a_iS_i(n)-K \right)^+,
$$
where $a_i\geq 0$, $\sum_i a_i=1$. The argument inside the
function $h$ is clearly a polynomial in $S_i(n)$ with non-negative coefficients
and hence Lemma \ref{L:f-supermodular} applies. Consequently, $X$ is fibrewise supermodular.
Observe that the same conclusion follows from \cite[Proposition~2.2.6]{MR3154633}

Evaluating formula \eqref{Eq:cmax-path-ind} for the upper bound of the no-arbitrage contingent claim price interval and the corresponding formula
for the lower bound we obtain
\begin{align} 
C_{\max}(v)=
\sum_{k_0+\ldots+k_m=n-k} &\frac{n!}{k_0!\cdots k_m!} \B P(\mu_0)^{k_0}\cdots\B P(\mu_m)^{k_m}
\times
\nonumber \\
\times & \left(  
\sum_{i=1}^m a_iD_i^{d_v+k_0+\ldots+k_{i-1}}U_i^{u_v+k_i+\ldots+k_m}
-K\right)^+
\label{E:max_basket_call} 
\end{align} 
\begin{align} \label{E:min_basket_call}
C_{\min}(v)=
\sum_{k_0+\ldots+k_m=n-k} &\frac{n!}{k_0!\cdots k_m!} \B P'(\nu_0)^{k_0}\cdots\B P'(\nu_m)^{k_m}
\times
\nonumber \\
\times & \left(  
\sum_{i=1}^m a_iD_i^{d_v+k_0+\ldots+k_{i-1}}U_i^{u_v+k_i+\ldots+k_m}
-K\right)^+,
\end{align}
where $d_v,u_v$ are chosen so that $D_i^{d_v}U_i^{u_v}$ corresponds to the
vertex $v\in \B T_k$. The formula for the minimum is correct under the assumption
that $\sum_i b_i\leq 1$.

By replacing $h$ with any other convex function we obtain
an analogous formula for a more general contingent claim.
\end{example}

\begin{example}[European basket put option]\label{E:euro-put}
Let $h\colon \BB R\to \BB R$ be a convex function. Consider a random variable
$$
X = h\left( K - \sum_i a_i S_i(n) \right),
$$
where $K,a_i\geq 0$ and $\sum a_i = 1$. If $h(x)=x^+$ then we obtain the standard
European basket put option. Notice that if $h$ is convex then so is
$x\mapsto h(K-x)$. Thus by restricting the function $- \sum_i a_i S_i(n)$
to the elements of the sample space with all but the $k$-th column fixed we obtain
a polynomial in $\psi_i(k)$ with non-positive coefficients.  It follows from
the version of Proposition 2.2.6 (a) in \cite{MR3154633} with coefficients
$a_i\leq 0$ (the proof is analogous) that $X$ is fibrewise supermodular and
hence Theorem \ref{T:assaf} applies. The formulas are similar to the ones in
the previous example.
\end{example}

\begin{example}[Arithmetic average Asian basket call or put option]\label{E:asian-call}
Let 
$$
X = \left(\frac{1}{n}\left( \sum_i a_{1i}S_i(1) + \cdots +\sum_i a_{ni}S_i(n) \right)-K\right)^+,
$$
where $a_{ki}\geq 0$ and $\sum_i a_{ki}=1$ for each $k$. It is thus an Asian
basket call option.
It follows directly form Lemma \ref{L:f-supermodular} that $X$ is
fibrewise supermodular and the upper bounds of the no-arbitrage values of the contingent
claim $X$ can be computed with the formula \eqref{Eq:cmax-better} and the analogous
one for the lower bounds under the required assumption.
Similarly a put contingent claim, obtained by negating the function inside $(\ )^+$,
is also fibrewise supermodular.
\end{example}

\subsection{Fibrewise submodular contingent claims}\label{SS:f-submodular-payoff}
In this section we present for completeness a result for fibrewise submodular contingent claims.
Since submodular contingent claims are rare we omit the proof which
is a straightforward adaptation of the proof for the fibrewise supermodular
case starting with Theorem \ref{T:min submodular n step}. 
Recall that a function $f$ is {\em submodular} if $-f$ is supermodular.

\begin{theorem}\label{T:assaf-sub}
Let $\B P$ and $\B P'$ be the $1$-step martingale probability measures from
Theorem~\ref{T:assaf} (described precisely in Section \ref{SS:max-measure} and
\ref{SS:minimal}).  Let $X\colon \Omega_n\to \BB R$ be a fibrewise submodular
contingent claim. Then the lower bound of a no-arbitrage price interval for $X$ at a vertex $v\in \B T_k$
is given by
$$
C_{\min}(v) = \B E_{\B P^n}(X\ |\ \C F_k)(v).
$$
If $\sum_i b_i\leq 1$ then the upper bound of the no-arbitrage price interval for $X$ at a vertex $v\in \B T_k$ is
equal to
$$
C_{\max}(v) = \B E_{(\B P')^n}(X\ |\ \C F_k)(v).
$$
\end{theorem}

\begin{example}[Geometric average Asian call]\label{E:asian-geom}
Let $S(k) = \sum_i a_{ki}S_i(k)$. Consider an Asian call option based on geometric
mean. Its pay-off is given by
$$
X = \left( \sqrt[n]{S(1)\cdots S(n)}-K \right)^+.
$$
This pay-off is neither super- nor submodular. However, the Arithmetic Mean -- Geometric Mean
Inequality (and Example \ref{E:asian-call}) yields an upper bound of the no-arbitrage price values.

For scenarios such that $\sqrt[n]{S(1)\cdots S(n)}\geq K$ we have that \newline 
$X= \sqrt[n]{S(1)\cdots S(n)}-K$ and hence, up to an additive constant, $X$ is
a composition of a polynomial $S(1)\cdots S(n)$ with a concave function $\sqrt[n]{\phantom{x}}$.
Thus Theorem \ref{T:assaf-sub} applies and the bounds of the no-arbitrage price interval
can be computed effectively {\it for some} vertices $v\in \B T$.
\end{example}

\subsection{Examples where the extremal martingale measures are not product}
\label{SS:nonexamples}
Consider a $1$-step model with assets $S_1,S_2\colon \Omega \to \BB R$
each following a binomial model with respective price ratios
$0<D_1,D_2<R<U_1,U_2$, where $R$ is the risk-free rate. Let
$p_i = \frac{R-D_i}{U_i-D_i}$ be the risk neutral measure for
each of the asset in its own binomial model. The equations
\eqref{eq:mart_meas_n=1} have the form
\begin{eqnarray*}
U_1q_1+ U_1q_2 + D_1q_3 + D_1q_4&=&R \\
U_2q_1+ D_2q_2 + U_2q_3 + D_2q_4&=&R \\
q_1+q_2+q_3+q_4   &=&1\nonumber\\
q_j &\geq& 0.\nonumber
\end{eqnarray*}
It is a straightforward computation that the solution set is the 
interval of points of the form
$$
Q(t)=(t,p_1-t,p_2-t,1-p_1-p_2+t)\in \BB R^4,
$$
where $\max\{p_1+p_2-1,0\}=t_{\min} \leq t \leq t_{\max}=\min\{p_1,p_2\}$. 
In particular, for any of the choices made the interval is parallel to the
vector $(1,-1,-1,1)$. 

Let $X=(X_1,X_2,X_3,X_4)$ be the contingent claim. It follows that 
$C_{\min}(0) = \langle X,Q(t_{\min})\rangle$
and
$C_{\max}(0) = \langle X,Q(t_{\max})\rangle$
provided that $\langle X,(1,-1,-1,1)\rangle = X_1-X_2-X_3+X_4>0$. The opposite inequality
implies that that maximal and the minimal values are attained at $Q(t_{\min})$
and $Q(t_{\max})$, respectively. This observation is the used in the following example.

\begin{example}\label{E:spread}
Let $m=2$ and $n=2$. Consider two assets $S_1$ and $S_2$ with the following initial data: $S_1(0)=100, U_1=1.2, D_1=0.8, K_1=100, S_2(0)=90, U_1=1.15, D_1=0.9, K_2=110.$
Consider a spread 
$$
X = \left( \frac{S_1(2)+S_2(2)}{2} - K_1  \right)^+ - \left(\frac{S_1(2)+S_2(2)}{2}-K_2  \right)^+
$$
If the prices of both assets go up in the first step then at time $t=2$ the values of $X$ are
given by
$$
X\left( 
\begin{smallmatrix}
1 & *\\
1 & *\\
\end{smallmatrix}
 \right)
= (10,10,7.5125,0)
$$
and when at time $t=1$ the first asset goes up and the second down we have
$$
X\left( 
\begin{smallmatrix}
1 & *\\
0 & *\\
\end{smallmatrix}
 \right)
= (10,8.45,0,0).
$$
The above are straightforward calculations.
The inner product with $(1,-1,-1,0)$ of the first one is negative while of the second one is positive.
This shows that the single-step maximal martingale measures corresponding to these two conditional situations
are distinct and hence the maximal martingale measure cannot be product of the same single-step measure.
\end{example}

\appendix
\section{}\label{S:appendix}

\subsection*{Notions and notation}

Let $\Omega$ be a finite sample space and let $\Delta(\Omega)$ be the set of all probability measures on $\left(\Omega, 2^{\Omega}\right)$. Each probability measure in $\Delta(\Omega)$ can be identified with its probability function (as was done throughout this paper), so the set $\Delta(\Omega)$ is defined as follows:
\[
\Delta(\Omega) = \left\{ x \colon \Omega \to \BB R \ : \ \sum_{\omega \in \Omega} x(\omega)=1\,, x(\omega) \geq 0\right\}.
\]
This is the standard simplex in the space $\BB R^N$, where $N$ is the number of elements in $\Omega$.
Given a probability measure $x\in\Delta(\Omega)$, any function $f\colon \Omega \to \BB R$ is viewed as a random variable on $\left(\Omega, 2^{\Omega}, x\right)$.

We will denote by $\B E_{x}(f)$ the expected value of $f$ with respect to the probability measure $x\in\Delta(\Omega)$.

For any $\omega \in \Omega$ we denote by $e_\omega \colon \Omega \to \{0,1\}$ the following indicator function:  
\[
e_\omega (\omega') =
\left\{
\begin{array}{ll}
1  & \text{if $\omega' = \omega$} \\
0  & \text{if $\omega' \neq \omega$}
\end{array}\right. ,
\]
for any $\omega' \in \Omega$. 
We will write $\B 1$ for the constant function on $\Omega$: $\B 1(\omega)=1$ for any $\omega \in \Omega$.

Let $A_1,\dots,A_n$ be finite sets and let $f_i \colon A_i \to \BB R$ be functions. Define the function $f_1 \otimes \cdots \otimes f_n\colon A_1 \times \cdots \times A_n \to \BB R$ by
\[
(f_1 \otimes \cdots \otimes f_n)(a_1,\dots,a_n) = f_1(a_1) \cdots f_n(a_n),
\]
where $a_k\in A_k$.

\subsection*{Single-step Bernoulli trials}
Let us fix $m \geq 1$ and denote $\C L=2^{\{1,\dots,m\}}$, the power-set of $\{1,\dots,m\}$.
The set $\C L$ is the natural sample space for $m$ single-step Bernoulli trials. Each set $S\in \C L$ consists of numbers that correspond to trials in which a ``success'' occurred. For example, the set $S=\{2,5,6\}$ corresponds to the scenario where ``success'' occurred in trials 2,5, and 6, and ``failure'' occurred in the rest of the trials; the empty set $S=\{\emptyset\}$ corresponds to the scenario where ``failure'' occurred in all trials, etc.

\begin{remark}\label{R:L-to-Omega}
The set $\C L$ can be put into one-to-one correspondence with the sample space $\Omega_1$ of a single-step market model with $m$ assets (see Section \ref{S:single-step}). Recall that $\Omega_1$ consists of $(m\times 1)$-matrices with binary coefficients and the number of elements in $\Omega_1$ is $N=2^m$. 
\end{remark}

Let us introduce a set of random variables $\ell_i \colon \C L \to \{0,1\}$,  $i=1, \dots , m$ as follows:
\[
\ell_i(S) =
\left\{
\begin{array}{ll}
 1 & \text{if $ i \in S$} \\
 0 & \text{if $ i \notin S$}
\end{array}\right. ,
\]
where $S\in \C L$. The random variable $\ell_i$ identifies scenarios $S\in \C L$ in which ``success'' occurred in the $i$-th trial. In other words, $\ell_i(S)$ represents the result of the $i$-th trial in scenario $S$, where $\ell_i(S)=1$ corresponds to  ``success'' and $\ell_i(S)=0$ corresponds to  ``failure''. For simplicity, we will call random variables $\ell_1, \dots, \ell_m$ the {\em Bernoulli trials}.  

\begin{remark}\label{R:l_to_psi}
The Bernoulli trials $\ell_i \colon \C L \to \{0,1\}$,  $i=1, \dots , m$ can be put into one-to-one correspondence with the stock price ratios $\psi_i(1) \colon \Omega_1 \to \{D_i,U_i\}$,  $i=1, \dots , m$ (see Section \ref{S:single-step}). Recall that 
\[
\psi_i(1)(\omega) =
\begin{cases}
	D_i & \text{ if } \omega_{i1}=0\\
	U_i & \text{ if } \omega_{i1}=1.\\
\end{cases}
\]
It follows that 
\begin{equation}\label{E:l_to_psi}
\psi_i(\omega) = (U_i-D_i)\ell_i(S) + D_i.
\end{equation}
(here we used the simplified notation $\psi_i(\omega)$ to denote $\psi_i(1)(\omega)$). In (\ref{E:l_to_psi}), the set $S\in \C L$ corresponds to the sample space element $\omega \in \Omega_1$. For example, for $m=3$, $S=\{1,2\}$ corresponds to $\omega = (1,1,0)^T$. 
\end{remark}

It would be convenient to write $\ell_0 = \B 1$ and $\mu_i = \{1,\dots,i\} \in \C L$, for every $i=0,\dots,m$. Notice that $\mu_0=\emptyset$.
Also observe that for all $0 \leq j \leq m$
\[
\ell_i(\mu_j) = \left\{
\begin{array}{ll}
1 & \text{if $i \leq j$} \\
0 & \text{if $i >j$}.
\end{array}\right.
\]
This identity holds for $i=0$ as well as for all $i=1,\dots,m$.
We will also denote $\nu_0=\emptyset$ and $\nu_i=\{i\}$.
Clearly if $1 \leq i \leq m$ then $\ell_i(\nu_j)=\delta_{i,j}$.

\subsection*{Single-step martingale measures}

Fix $0 \leq b_1,\dots,b_m \leq 1$ and denote $\B b=(b_1,\dots,b_m)$.
Define a subset ${\rm M}_1(\B b)\subset \Delta (\C L)$ as follows: 
\begin{equation}\label{E:P_b}
{\rm M}_1(\B b) = \{ x \in \Delta(\C L) \ :  \ \B E_x(\ell_i)=b_i \text{ for all $i=1,\dots,m$}\}.
\end{equation}
Notice that ${\rm M}_1(\B b)$ is a polyhedral set in $\BB R^N$, where $N=2^m$. 
\begin{remark}\label{R:P_b_isM_1}
If $b_i=\frac{R-D_i}{U_i-D_i}$, $i=1, \dots ,m$, the set ${\rm M}_1(\B b)$ coincides with the set $M_1$ of martingale measures on $\Omega_1$ (see Section \ref{S:single-step}). Indeed, consider system (\ref{eq:mart_meas_n=1}) which defines a martingale measure $\B P=(p_1,p_2,\ldots,p_N)$ on $\Omega _1$ and replace each $p_i$ with $x_i$. Further, use (\ref{E:l_to_psi}) to express each $\psi_i(\omega_j)$ in terms of $\ell_i(S_j)$, where $S_j$ corresponds to $\omega_j$. After some straightforward algebraic transformations, (\ref{eq:mart_meas_n=1}) becomes:
\begin{eqnarray}\label{E:P_b_isM_1}
\sum_{j=1}^N x_j\ell_i(S_j)&=&\frac{R-D_i}{U_i-D_i}\\
\sum_{j=1}^N x_j&=&1\nonumber \\
x_i&\ge& 0,\nonumber
\end{eqnarray} 
for $i=1, \dots ,m$. This system of equations and inequalities is equivalent to the definition (\ref{E:P_b}) of the set ${\rm M}_1(\B b)$ for the case
\begin{equation}\label{E:b_i}
b_i=\frac{R-D_i}{U_i-D_i},
\end{equation}  
 for $i=1, \dots ,m$. 	
\end{remark}

\subsection*{Optimization problems 1 and 2}
Given a random variable $f \colon \C L \to \BB R$ we will consider the following optimization problems:

{\bf Problem 1.} {\it Find $q^* \in {\rm M}_1(\B b)$ such that}  
\[
\B E_{q^*}(f)=\max \{\B E_p(f) : p \in {\rm M}_1(\B b)\}.
\]

{\bf Problem 2.} {\it Find $q_* \in {\rm M}_1(\B b)$ such that}  
\[
\B E_{q_*}(f)=\min \{\B E_p(f) : p \in {\rm M}_1(\B b)\}.
\]
\begin{remark}\label{R:opt_prob1_and_2}
Let (\ref{E:b_i}) hold. Then the above optimization Problem 1 is equivalent to the LP problem (\ref{eq:LP}) of finding the upper limit of the no-arbitrage contingent claim price interval and the maximal martingale measure in a single-step market model with $m$ assets (see Section \ref{SS:simplex}). Similarly, Problem 2 is equivalent to the LP problem of finding the lower limit of the no-arbitrage contingent claim price interval and the minimal martingale measure. 
\end{remark}
In what follows, we will always impose the following

{\bf Assumption.} The vector $\B b$ is decreasing, namely
\begin{equation}\label{Eq:decreasing}
b_1 \geq \cdots \geq b_m.
\end{equation}
Recall that for a fixed $i$, $b_i=\B E_x(\ell_i)$, where $x$ is any probability measure from the set ${\rm M}_1(\B b)$ and thus the condition (\ref{Eq:decreasing}) is easily fulfilled by permuting the Bernoulli trials  $\ell_i$, $i=1, \dots, m.$

It will be convenient to denote $b_0=1$.

Before we proceed with the solutions of the above optimization problems, we need to introduce some necessary definitions.

\subsection*{Supermodular vertex measures}
Given the setup above define for $i=0,\dots,m$: $q^{(i)}=b_i-b_{i+1}$ and  $q^{(m)}=b_m$. Notice that $q^{(0)}=1-b_1$.
Define a function $q^* \colon \C L \to \BB R$ by
\begin{equation}\label{Eq:supervertex}
q^* = \sum_{i=0}^m q^{(i)} \cdot e_{\mu_i}.
\end{equation}
Notice that $q^{(i)} \geq 0$ since $\B b$ is decreasing (see  (\ref{Eq:decreasing})) and $0 \leq b_i \leq 1$. Therefore, $q^* \geq 0$.
Also
\[
\sum_{S \in \C L} q^*(S) = \sum_{i=0}^m q^{(i)} = (1-b_1)+\left(\sum_{i=1}^{m-1} b_i-b_{i+1}\right) + b_m = 1.
\]
So $q^*$ is a probability measure on $\C L$.
Also,
\[
\B E_{q^*}(\ell_i) = \sum_{S \in \C L} q^*(S) \ell_i(S) 
= \sum_{j=0}^m q^{(j)} \ell_i(\mu_j) = \sum_{j=i}^m q^{(j)} 
= \sum_{j=i}^{m-1} (b_i-b_{i+1}) + b_m = b_i
\]
for all $i=0,\dots,m$.
We deduce that $q^* \in {\rm M}_1(\B b)$.
We call $q^*$ the {\em upper supermodular vertex measure} of ${\rm M}_1(\B b)$.
Indeed, it is a vertex of this polyhedral set, although we will not use this fact directly.

Define for all $i=0,\dots,m$: $q_{(0)} = 1-\sum_{i=1}^m b_i$ and $q_{(i)}=b_i$. Define a function $q_* \colon \C L \to \BB R$ by
\begin{equation}
q_* = \sum_{i=0}^m q_{(i)} \cdot e_{\nu_i}.
\label{Eq:subvertex}
\end{equation}
Since $b_i \geq 0$ it is clear that if $\sum_{i=1}^m b_i \leq 1$ then $q_* \geq 0$.
It is also clear that
\[
\sum_{S \in \C L} q_*(S) = \sum_{j=0}^m q_{(j)} = 1.
\]
Hence, $q_*$ is a probability measure on $\C L$.
Moreover, for any $i=1,\dots,m$
\[
\B E_{q_*}(\ell_i) = \sum_{S \in \C L} \ell_i(S) \cdot q_*(S) = \sum_{j=0}^m \ell_i(\nu_j) \cdot q_{(j)} = b_i
\]
since $\ell_i(\nu_j)=\delta_{i,j}$.
We deduce that $q_* \in {\rm M}_1(\B b)$ provided $\sum_{i=1}^m b_i \leq 1$.
We call $q_*$  the {\em lower supermodular vertex measure} of ${\rm M}_1(\B b)$.
Indeed, it is a vertex of ${\rm M}_1(\B b)$, but we will not use this fact directly.

\subsection*{Supermodular functions}\label{SS:supermodular}
Recall (see e.g. \cite[Section 2]{MR3154633}) that a function $f \colon \C L \to \BB R$ is called {\em supermodular} 
if for any $S,T \in \C L$ 
\[
f(S \cup T) + f(S \cap T) \geq f(S) + f(T).
\]
We observe that the functions $\ell_1,\dots,\ell_m$ give rise to an injective function 
\[
(\ell_1,\dots,\ell_m) \colon \C L \to \BB R^m
\]
whose image is the set $\{0,1\}^m$ of the vertices of the cube $[0,1]^m$.
This gives a convenient way to define several important supermodular functions.
See \cite[Section 2]{MR3154633} and \cite{MR468177,MR548704} for examples and
properties of supermodular functions.

\subsection*{Solving optimization Problems 1 and 2}

\begin{theorem}\label{T:supermodular vertices min max}
Suppose that $\B b$ is decreasing (see (\ref{Eq:decreasing})).
\begin{enumerate}[label=(\roman*)]
\item 
\label{T:modular:max}
 For any supermodular function $f \colon \C L \to \BB R$
\[
\max \{\B E_x(f) : x \in {\rm M}_1(\B b)\} =\B E_{q^*}(f).
\]
\item
\label{T:modular:min}
Suppose that $\sum_{i=1}^m b_i \leq 1$.
Then for any supermodular function $f \colon \C L \to \BB R$
\[
\min \{\B E_x(f) : x \in {\rm M}_1(\B b)\} =\B E_{q_*}(f).
\]
\end{enumerate}
That is, the maximum expectation over ${\rm M}_1(\B b)$ of any supermodular function is always attained at the upper supermodular vertex measure and the minimum is attained at the lower supermodular vertex measure, if the latter is defined.
\end{theorem}

\begin{proof}
\ref{T:modular:max}
Define for $i=0,\dots,m$: $\alpha_0 = f(\mu_0)$ and $\alpha_i=f(\mu_i)-f(\mu_{i-1})$.

Define a function $g \colon \C L \to \BB R$ by
\[
g=\sum_{i=0}^m \alpha_i \cdot \ell_i.
\]
Observe that for any $i=0,\dots,m$
\[
g(\mu_i) = \sum_{j=0}^m \alpha_j \cdot \ell_j(\mu_i) = \sum_{j=0}^i \alpha_j = f(\mu_i).
\]
Set $h=g-f$.
We claim that $h \geq 0$, namely $h(S) \geq 0$ for all $S \in \C L$.
Suppose this is false.
Among all $S \subseteq \{1,\ldots,m\}$ for which $h(S)<0$, choose one with the largest possible $k$ such that $\mu_k=\{1,\dots,k\} \subseteq S$.
Clearly $k<m$ since we know that $h(\mu_m)=0$.
Also, $k+1 \notin S$.
Set $T=S\cup\{k+1\}$.
Then $\mu_{k+1} \subseteq T$ and $S\setminus \mu_k = T \setminus \mu_{k+1}$.
Therefore
\begin{eqnarray*}
&& h(S) = g(S)-f(S) = \alpha_0+\sum_{i=1}^k \alpha_i + \sum_{i=k+2}^m \alpha_i \ell_i(S) - f(S) \\
&& h(T) = g(T)-f(T) = \alpha_0+\sum_{i=1}^{k+1} \alpha_i + \sum_{i=k+2}^m \alpha_i \ell_i(S) - f(T).
\end{eqnarray*}
Observe that $T=S \cup \mu_{k+1}$ and that $S \cap \mu_{k+1} = \mu_k$.
Since $f$ is supermodular,
\[
h(S)-h(T) = f(T)-f(S)-\alpha_{k+1} = f(S \cup \mu_{k+1})-f(S) - f(\mu_{k+1}) + f(\mu_k) \geq 0.
\]
In particular $h(T) \leq h(S)<0$ which is a contradiction to the maximality of $k$.

We will complete the proof by showing that $\B E_{q^*}(f) \geq \B E_{x}(f)$ for any $x \in {\rm M}_1(\B b)$.
First, by definition of ${\rm M}_1(\B b)$ and the linearity of the expectation for any $x \in {\rm M}_1(\B b)$ we have
\[
\B E_x(g) = \alpha_0+\sum_{i=1}^m \alpha_i b_i.
\]
Next, we compute
\begin{align*}
	\B E_{q^*}(f) &= 
	\sum_{S \subseteq \{1,\dots,m\}} f(S) q^*(S)
	=
	\sum_{i=0}^m f(\mu_i) q^{(i)}
	=\sum_{i=0}^m f(\mu_i) (b_i-b_{i+1}) + f(\mu_m) b_m \\
	&=b_0 f(\mu_0) + \sum_{i=1}^m (f(\mu_i)-f(\mu_{i-1}) b_i
	=\alpha_0 + \sum_{i=1}^m \alpha_i b_i =
	\B E_{q^*}(g).
\end{align*}
Now, since $h \geq 0$, for any $x \in {\rm M}_1(\B b)$
\[
\B E_x(f) = \B E_x(g-h) \leq \B E_x(g) = \alpha_0 \B E_x(\B 1) + \sum_{i=1}^m \alpha_i \B E_x(\ell_i) = \alpha_0 + \sum_{i=1}^m \alpha_i b_i = \B E_{q^*}(g),
\]
and the proof is complete.

\ref{T:modular:min}
Define for every $i=0,\dots,m$: $\alpha_0=f(\nu_0)$ and $\alpha_i=f(\nu_i)-f(\nu_0)$.
Define $g \colon \C L \to \BB R$ by
\[
g=\alpha_0 \ell_0 + \sum_{i=1}^m \alpha_i \ell_i.
\]
Set $h=f-g$.
We claim that $h \geq 0$.
Assume that this is false.
Choose $S \subseteq \{1,\ldots,m\}$ of the smallest possible cardinality such that $h(S)<0$.
Clearly $S \neq \emptyset$ since $g(\nu_0) = \alpha_0=f(\nu_0)$, so $h(\nu_0)=0$ (and $\nu_0=\emptyset$).
Choose some $k \in S$ and set $T=S \setminus \{k\}$.
Then 
\begin{eqnarray*}
&& h(S)=f(S)-g(S) = f(S) - \alpha_0-\sum_{i\in S} \alpha_i =
f(S)-\alpha_0-\alpha_k - \sum_{i \in T} \alpha_i \\
&& h(T) = f(T) -g(T) = f(T) - \alpha_0-\sum_{i \in T} \alpha_i.
\end{eqnarray*}
Since $f$ is supermodular
\[
h(S)-h(T) = f(S)-f(T)-\alpha_k =
f(T \cup \{k\}) -f(T) - f(\{k\})-f(\emptyset) \geq 0.
\]
It follows that $h(T) \leq h(S)<0$, contradiction to the minimality of $|S|$.

It remains to show that $\B E_x(f) \geq \B E_{q_*}(f)$ for any $x \in {\rm M}_1(\B b)$.
By definition, for any $x \in {\rm M}_1(\B b)$ we have $\B E_x(g)=\alpha_0+\sum_{i=1}^m \alpha_i b_i$.
Under the hypothesis $\sum_{i=1}^m b_i \leq 1$ we have $q_* \in {\rm M}_1(\B b)$, so
\begin{multline*}
\B E_{q_*}(f) = \sum_{S \subseteq \{1,\ldots,m\}} f(S) \cdot q_*(S) = 
\sum_{i=0}^m f(\nu_i) \cdot q_{(i)} = (1-\sum_{i=1}^m b_i) f(\nu_0) + \sum_{i=1}^m f(\nu_i) b_i 
\\
=f(\nu_0) + \sum_{i=1}^m (f(\nu_i)-f(\nu_0))b_i = 
\alpha_0 + \sum_{i=1}^m \alpha_ib_i = \B E_{q_*}(g).
\end{multline*}
Next, consider any $x \in {\rm M}_1(\B b)$.
Since $h\geq 0$, we clearly have 
$
\B E_x(f) = \B E_x(g+h) \geq \B E_x(g) = \B E_{q^*}(g).
$
This completes the proof.
\end{proof}

\begin{remark}
Part \ref{T:modular:max} of Theorem \ref{T:supermodular vertices min max} is
tightly related to Lov\'asz extensions \cite{MR717403}.  Indeed,
$E_{q^*}(f)=(-f)^L(\B b)$ where on the right hand side is the Lov\'asz
extension of $-f$.

We identify $m$ with the set $\{1,\dots,m\}$ and $\C L=\{0,1\}^m$ with its
power set.  Consier a function $f \colon \{0,1\}^m \to \BB R$.
The {\em convex closure} of $f$ is the function $f^-\colon [0,1] \to \BB R$
defined by 
\[
f^-(\B x) = \min\left\{ \sum_{S \subseteq m} \alpha_S f(S) : \sum_{S\subseteq m} \alpha_S 
\cdot \mathbf{1}_S = \B x, \sum_{S \subseteq m} \alpha_S=1, \alpha_S \geq 0 \right\}.
\]
Thus, by definition, given $\B b \in [0,1]^m$ the value of $f^{-}(\B b)$ is the
minimum of $E_\alpha(f)$ over all probability measure s $\alpha$ on $\C L$ for
which $E_\alpha(\ell_i)=x_i$ where $\ell_i \colon \C L \to \{0,1\}$ are
projections to the $i$-th factor, namely $\alpha \in P(\B b)$.

The {\em Lov\'asz extension} of $f$  is the function $f^L \colon [0,1]^m \to \BB R$ defined as follows.
Given $x \in [0,1]^m$ write $m=\{k_1,\dots,k_m\}$ where $x_{k_1} \geq x_{k_2} \geq \cdots \geq x_{k_m}$.
For any $0 \leq i \leq m$ set $S_i=\{k_1,\dots,k_i\} \subseteq m$.
Then there are unique $\lambda_0,\dots,\lambda_m \geq 0$ such that $\sum_i \lambda_i =1$ and 
$x=\sum_{i=0}^m \lambda_i \cdot \mathbf{1}_{S_i}$.
We define
\[
f^L(x) = \sum_{i=0}^m \lambda_i f(S_i).
\]
In the notation of this paper, if $x=\B b$ then $\mathbf{1}_{S_i}=\mu_i$ and $\lambda_i= q^{(i)}$.
Thus, $f^L(b)=E_{q^*}(f)$.
It is well known that $f^-=f^L$ if and only if $f$ is submodular \cite{VondrakNotes}, and notice that $f$ is submodular 
iff $-f$ is supermodular.
Thus, $E_{q^*}(f)$ is the maximum expectation of $f$ with respect to probability measures $\alpha \in P(\B b)$.
\end{remark}

\subsection*{Multi-step Bernoulli trials}
Fix some $n \geq 1$ and consider $\C L^n$.
This is the natural sample space for $n$ iterations of $m$ single-step Bernoulli trials, or, equivalently, for $m$ $n$-step Bernoulli trials.

\begin{remark}\label{R:L-n-to-Omega-n}
	The set $\C L^n$ can be put into one-to-one correspondence with the sample space $\Omega_n$ of the $n$-step market model with $m$ assets (see Section \ref{S:preliminaries}). Recall that $\Omega_n$ consists of $(m\times n)$-matrices with binary coefficients and the number of elements in $\Omega_n$ is $N=2^{mn}$. 
\end{remark}

For every $1 \leq k\ \leq n$, define the set of random variables $\ell_i^k:\C L^n \to \{0,1\} $, $i=1, \dots ,m$ as follows:
\[
\ell_i^k \ \overset{\text{def}}{=} \ \underbrace{\B 1 \otimes \cdots \otimes \B 1}_{\text{$k-1$ times}} \otimes\ \ell_i \otimes \underbrace{\B 1 \otimes \cdots \otimes \B 1}_{\text{$n-k-1$ times}}
\] 
Notice that $\ell_i^k\left(S_1,\dots ,S_n\right)$ represents the result of the $i$-th trial in the $k$-th iteration according to scenario $\left(S_1,\dots ,S_n\right)$: $\ell_i^k\left(S_1,\dots ,S_n\right)=\ell_i(S_k)$.

Further, for every $1 \leq k\ \leq n$, define a vector random variable \newline $L^k=(\ell_1^k,\dots,\ell_m^k):\C L^n \to \{0,1\}^m $. Notice that $L^k\left(S_1,\dots ,S_n\right)$ represents the result of the $k$-th iteration of $m$ Bernoulli trials according to scenario $\left(S_1,\dots ,S_n\right)$.

\subsection*{Multi-step martingale measures}

We fix some $\B b=(b_1,\dots,b_m)$ as above.

Define a subset ${\rm M}_n(\B b)\subset \Delta (\C L^n)$ as follows: 
\begin{equation}\label{E:M_n_b}
	{\rm M}_n(\B b) = \{ p \in \Delta(\C L^n) \ :  \ \B E_p\left(L^k\ |\ \C F_{k-1}\right)=\B b, \text{ for all $k=1,\dots,n$}\},
\end{equation}
where $\C F_{k-1}$ is a $\sigma$-algebra generated by the random vectors $L^1, L^2,\ldots,L^{k-1}$. 
We will denote the above conditional expectation by 
$\B E_{p}\left( L^k\ |\ L^1,\ldots,L^{k-1} \right)$.

Notice that ${\rm M}_n(\B b)$ is a polyhedral set in $\BB R^N$ given by the following equations: 
\begin{align}
	\label{E:def RN}
\sum_{\tau_k,\dots,\tau_n \in \C L} &\ell_i(\tau_k)\cdot x(\lambda_1,\dots,\lambda_{k-1},\tau_k,\dots,\tau_n)\\ 
&= b_i \cdot \sum_{\tau_k,\dots,\tau_n \in \C L} x(\lambda_1,\dots,\lambda_{k-1},\tau_k,\dots,\tau_n) 
\nonumber
\end{align}
where $x \in \Delta(\C L^n)$ and 
$1 \leq k \leq n, \ \lambda_1,\dots,\lambda_{k-1} \in \C L, \ i=1,\dots,m)$.
To see this divide both sides of (\ref{E:def RN}) by the sum in the right hand side and observe that the
resulting left hand side is exactly the required conditional expectation.
Notice that the product probability measure $q \otimes \cdots \otimes q$ belongs to ${\rm M}_n(\B b)$ 
for $q\in {\rm M}_1(\B b)$.

\begin{remark}\label{R:M_n_b_isM_n}
If \eqref{E:b_i} holds for $i=1, \dots ,m$, the set ${\rm M}_n(\B b)$ coincides with the set ${\rm M}_n$ of martingale measures on $\Omega_n$ (see \eqref{Eq:martingale_N} and \eqref{Eq:martingale_N_jump} ).
\end{remark}

\subsection*{Optimization problems 3 and 4}

Given a random variable $f \colon \C L^n \to \BB R$, we will consider the following optimization problems:

{\bf Problem 3.} {\it Find $p^* \in {\rm M}_n(\B b)$ such that}  
\[
\B E_{p^*}(f)=\max \{\B E_p(f) : p \in {\rm M}_n(\B b)\}.
\]

{\bf Problem 4.} {\it Find $p_* \in {\rm M}_n(\B b)$ such that}  
\[
\B E_{p_*}(f)=\min \{\B E_p(f) : p \in {\rm M}_n(\B b)\}.
\]

\begin{remark}\label{R:opt_prob3_and_4}
	Let (\ref{E:b_i}) hold. Then the above optimization Problem 3 is equivalent to the problem of finding the upper limit of the no-arbitrage contingent claim price interval at time zero and the maximal martingale measure in the $n$-step market model with $m$ assets (see Section \ref{SS:bounds}). Similarly, Problem 4 is equivalent to the problem of finding the lower limit of the no-arbitrage contingent claim price interval at time zero and the minimal martingale measure. 
\end{remark}

\subsection*{Fibrewise supermodular functions}\label{SS:fibrewise_supermodular}

\begin{definition}\label{D:f-supermodular}
	We say that $f \colon \C L^n \to \BB R$ is {\em fibrewise supermodular} if its restriction to the subset $(\lambda_1,\dots,\lambda_{k-1}) \times \C L \times (\lambda_{k+1},\dots,\lambda_n)$ is a supermodular function on $\C L$ for every $\lambda_1,\dots,\widehat{\lambda_k},\dots,\lambda_n \in \C L$.
\end{definition}

\begin{example}\label{Ex:fibrewise supermodular}
	Suppose that $u_{i,j} \colon \C L \to \BB R$ are affine functions, $1 \leq i \leq n$ and $1 \leq j \leq r$.
	Suppose that each $u_{i,j}$ has the form $\sum_k \alpha_k x_k+\gamma$ where $\alpha_k \geq 0$.
	Suppose that $h \colon \BB R \to \BB R$ is convex.
	Let  $g \colon \C L^n \to \BB R$ be the function $\sum_{j=1}^r u_{1,j} \otimes \cdots \otimes u_{n,j} +c \B 1$.
	Then $f=h \circ g|_{\C L^n}$ is fibrewise supermodular.
	
	Indeed, the restriction of $g$ to any subset $(\lambda_1,\dots,\lambda_{k-1}) \times \C L \times (\lambda_{k+1},\dots,\lambda_n)$ is a linear combination of the affine maps $u_{k,1}, \dots, u_{k,r}$ with non-negative coefficients and a multiple of~$\B 1$.
	Now appeal to \cite[Proposition 2.2.6 (a)]{MR3154633}.
\end{example}

\subsection*{Solving optimization Problems 3 and 4}

\begin{theorem}\label{T:min submodular n step}
	Suppose that $\B b$ is decreasing (see (\ref{Eq:decreasing})).
	\begin{enumerate}[label=(\roman*)]
		\item 
		\label{T:fibrewise_modular:max}
		For any fibrewise supermodular function $f \colon \C L^n \to \BB R$
		\[
	\max \{\B E_{p}(f) : p \in {\rm M}_n(\B b)\} =\B E_{p^*}(f) =\B E_{{q^*}^{\otimes n}}(f),
	\]
	where $q^*$ is the upper supermodular vertex measure defined in \eqref{Eq:supervertex}.
			\item
		\label{T:fibrewise_modular:min}
		Suppose $\sum_{i=1}^m b_i \leq 1$.  
		Then for any fibrewise supermodular function \newline $f \colon \C L^n \to \BB R$
	\[
	\min \{ \B E_p(f) : p \in {\rm M}_n(\B b)\} =\B E_{p_*}(f) = \B E_{{q_*}{\otimes n}}(f),
	\]
	where $q_*$ is the lower supermodular vertex measure defined in \eqref{Eq:subvertex}.
	\end{enumerate}
		Thus, $\B E_{p}(f)$ is maximized at the product measure $p^*={q^*}^{\otimes n}\in {\rm M}_n(\B b)$, where $q^*$ is the upper supermodular vertex measure (see \eqref{Eq:supervertex}).
		
		Further, $\B E_{p}(f)$ is minimized at the product measure $p_*={q_*}^{\otimes n}\in {\rm M}_n(\B b)$, where $q_*$ is the lower supermodular vertex measure (see \eqref{Eq:subvertex}), provided the latter is defined.
			
\end{theorem}

\begin{proof}
	We will show that if $x \in \OP{M}_n(\B b)$ then $\B E_x(f) \leq \B E_{{q^*}^{\otimes n}}(f)$ and $\B E_x(f) \geq \B E_{{q_*}^{\otimes n}}(f)$, thus proving the result since $q^{\otimes n} \in \OP{M}_n(\B b)$ for any $q \in {\rm M}_1(\B b)$.
	
	For every $0 \leq k \leq n$ and every $\lambda_1,\dots,\lambda_k \in \C L$ set
	\[
	y_k(\lambda_1 \dots \lambda_k) = \sum_{\theta_{k+1},\dots,\theta_n\in \C L} x(\lambda_1, \dots ,\lambda_k,\theta_{k+1},\dots,\theta_n).
	\]
	Thus, $y_k(\lambda_1 \dots \lambda_k)$ is the probabilty (with respect to $x$) of the event $\{L^1=\lambda_1,\dots,L^k=\lambda_k\}$.
	Define $x^{(k)} \colon \C L^n \to \BB R$ and $x_{(k)} \colon \C L^n \to \BB R$ by
	\begin{eqnarray*}
		x^{(k)}(\lambda_1,\dots,\lambda_n)&=&y_k(\lambda_1 \dots \lambda_k) \cdot q^*(\lambda_{k+1}) \cdots q^*(\lambda_n) \\
		x_{(k)}(\lambda_1,\dots,\lambda_n)&=&y_k(\lambda_1 \dots \lambda_k) \cdot q_*(\lambda_{k+1}) \cdots q_*(\lambda_n).
	\end{eqnarray*}
	Clearly, $x^{(k)} \geq 0$ and $x_{(k)} \geq 0$.
	Also, since $q^*$ is a probability measure on $\C L$
	\begin{align*}
	\sum_{\lambda_1,\dots,\lambda_n \in \C L} x^{(k)}(\lambda_1,\dots,&\lambda_n) 
	= \sum_{\lambda_1,\dots,\lambda_k} y_k(\lambda_1,\dots,\lambda_k) 
	\sum_{\lambda_{k+1},\dots,\lambda_n} \prod_{i=k+1}^n q^*(\lambda_i)\\ 
	&=\sum_{\lambda_1,\dots,\lambda_k \in \C L} y_k(\lambda_1,\dots,\lambda_k) =
	\sum_{\lambda_1,\dots,\lambda_n \in \C L} x(\lambda_1,\dots,\lambda_n) = 1.
\end{align*}
	So $x^{(k)} \in \Delta(\C L^n)$.
	An identical argument using  $q_* \in \Delta(\C L)$ gives $x_{(k)} \in \Delta(\C L)$.
	
	Clearly $x^{(n)}=x=x_{(n)}$ and $x^{(0)}={q^*}^{\otimes n}$ and $x_{(0)}={q_*}^{\otimes n}$.
	To complete the proof it remains to prove that $\B E_{x^{(k)}}(f) \geq \B E_{x^{(k+1)}}(f)$ for all $0 \leq k <n$, and similarly that $\B E_{x_{(k)}}(f) \leq \B E_{x_{(k+1)}}(f)$.
	
	If $y_k(\lambda_1,\dots,\lambda_k) >0$ then the map $p \colon \C L \to \BB R$
	\[
	p(\lambda) =\frac{y_{k+1}(\lambda_1,\dots,\lambda_k,\lambda)}{y_{k}(\lambda_1,\dots,\lambda_k)}
	\]
	is clearly a probability measure on $\C L$.
	Direct computation shows that for any $\lambda_1,\dots,\lambda_k \in \C L$
	\begin{eqnarray*}
		&& \sum_{\tau_{k+1},\dots,\tau_n \in \C L} \ell(\tau_{k+1})  x^{(k+1)}(\boldsymbol \lambda,\boldsymbol \tau) = \sum_\tau y_{k+1}(\boldsymbol \lambda \tau)\ell_i(\tau) \qquad \text{and} \\
		&& \sum_{\tau_{k+1},\dots,\tau_n \in \C L}  x^{(k+1)}(\boldsymbol \lambda,\boldsymbol \tau) = \sum_\tau y_{k}(\boldsymbol \lambda).
	\end{eqnarray*}
	The defining equations \eqref{E:def RN} of $\OP{M}_n(\B b)$ imply that 
	\[
	\sum_\tau y_{k+1}(\lambda_1,\dots,\lambda_k,\tau) \ell_i(\tau) = b_i \cdot y_k(\lambda_1,\dots,\lambda_k).
	\]
	Therefore, if $y_k(\lambda_1,\dots,\lambda_k)>0$ then $p \in {\rm M}_1(\B b)$.
	Hence, if $g \colon \C L \to \BB R$ is supermodular (resp. submodular) then $\B E_{q^*}(g) \geq \B E_p(g)$ (resp. $\B E_{q_*}(g) \leq \B E_p(g)$) which can be written explicitly
	\begin{eqnarray}
		\label{E:2}
		&& \sum_{\tau \in \C L} g(\tau) \cdot y_{k+1}(\lambda_1,\dots,\lambda_k,\tau) \leq y_{k}(\lambda_1,\dots,\lambda_k) \cdot \sum_{\tau \in \C L} g(\tau) \cdot q^*(\tau) \\
		\nonumber
		&& \sum_{\tau \in \C L} g(\tau) \cdot y_{k+1}(\lambda_1,\dots,\lambda_k,\tau) \geq y_{k}(\lambda_1,\dots,\lambda_k) \cdot \sum_{\tau \in \C L} g(\tau) \cdot q_*(\tau).
	\end{eqnarray}
	For $\theta_1,\dots,\theta_j \in \C L$ denote $q^*(\boldsymbol \theta)=q^*(\theta_1) \cdots q^*(\theta_j)$ and similarly for $q_*(\boldsymbol \theta)$.
	Now, 
	\begin{eqnarray*}
		&& \B E_{x^{(k)}}(f) = \sum_{\lambda_1,\dots,\lambda_k} 
		\sum_{\theta_{k+1},\dots,\theta_n} f(\boldsymbol \lambda \boldsymbol \theta) \cdot y_k(\boldsymbol \lambda) \cdot q^*(\boldsymbol \theta) \qquad \text{and} \\
		&& \B E_{x^{(k+1)}}(f) = 
		\sum_{\lambda_1,\dots,\lambda_{k+1}} \sum_{\theta_{k+2},\dots,\theta_n} f(\boldsymbol \lambda \boldsymbol \theta) \cdot y_{k+1}(\boldsymbol \lambda)\cdot q^*(\boldsymbol \theta) 
	\end{eqnarray*}
	Similar formulas hold for $\B E_{x_{(k)}}(f)$ and $\B E_{x_{(k+1)}}(f)$ by replacing $q^*$ with $q_*$ on the right hand sides.
	Since $f$ is fibrewise supermodular and $q^*(\boldsymbol \theta) \geq 0$ we can use the inequalities in \eqref{E:2} to continue the second equality:
	\begin{multline*}
		= \sum_{\lambda_1,\dots,\lambda_{k}} \sum_{\theta_{k+2},\dots,\theta_n} \sum_{\tau \in \C L} f(\boldsymbol \lambda \tau \boldsymbol \theta) \cdot y_{k+1}(\boldsymbol \lambda \tau ) \cdot q^*(\boldsymbol \theta)
		\\
		\leq 
		\sum_{\lambda_1,\dots,\lambda_{k}} \sum_{\theta_{k+2},\dots,\theta_n} \sum_{\tau \in \C L} f(\boldsymbol \lambda \tau \boldsymbol \theta) \cdot y_{k}(\boldsymbol \lambda) \cdot q^*(\tau) \cdot q(\boldsymbol \theta) 
		\\
		=
		\sum_{\lambda_1,\dots,\lambda_{k}} \sum_{\theta_{k+1},\dots,\theta_n} f(\boldsymbol \lambda \boldsymbol \theta) \cdot y_{k}(\boldsymbol \lambda) \cdot q^*(\boldsymbol \theta) = 
		\B E_{x^{(k)}}(f).
	\end{multline*}
	A similar calculation shows that $\B E_{x_{(k+1)}}(f) \geq \B E_{x_{(k)}}(f)$.
	This completes the proof.
\end{proof}

\par\bigskip\noindent
{\bf Acknowledgement.}
We thank both Bentley University and The University of Aberdeen for supporting our collaboration.


	
	
	
	


\bibliographystyle{amsplain}

\end{document}